\documentclass[ruled]{article}

\usepackage[utf8]{inputenc}  
\usepackage[T1]{fontenc}  

\usepackage{geometry}

\usepackage{color}
\usepackage{mathtools}
\usepackage{multibib}
\newcites{sup}{References}
\usepackage[authoryear]{natbib}
\usepackage{todonotes}
\usepackage{amsmath}
\usepackage{mathrsfs}
\usepackage{accents}
\usepackage{float}
\usepackage{graphicx}
\usepackage{caption}
\usepackage{subcaption}
\usepackage{enumitem}  
\usepackage[unicode=true,pdfusetitle, bookmarks=true,bookmarksnumbered=false,bookmarksopen=false, breaklinks=false,pdfborder={0 0 1},backref=false,colorlinks=true] {hyperref}
\usepackage{setspace}
\usepackage{amsthm}
\usepackage{bbm}
\usepackage{amssymb}

\hypersetup{linkcolor=blue,citecolor=blue}

\makeatletter

\renewcommand\paragraph{\@startsection{paragraph}{4}{\z@}%
            {-2.5ex\@plus -1ex \@minus -.25ex}%
            {1.25ex \@plus .25ex}%
            {\normalfont\normalsize\bfseries}}

\theoremstyle{plain}

\theoremstyle{plain}

\floatstyle{ruled}
\newfloat{algorithm}{tbp}{loa}
\providecommand{\algorithmname}{Algorithm}
\floatname{algorithm}{\protect\algorithmname}

\usepackage{algorithm}
\usepackage{dsfont}
\usepackage{algpseudocode}
\usepackage{mathtools}

\DeclareMathOperator*{\argmin}{argmin}

\newcounter{daggerfootnote}

\newcommand{\Natural}{\mathbb{N}}

\newcommand{\dd}{\mathrm{d}}

\newcommand{\R}{\mathbb{R}}
\renewcommand{\P}{\mathbb{P}}
\newcommand{\E}{\mathbb{E}}
\newcommand\iid{\stackrel{\mathclap{\normalfont\mbox{\tiny{iid}}}}{\sim}}

\newtheorem{remark}{Remark}
\newtheorem{theorem}{Theorem}
\newtheorem{lemma}{Lemma}
\newtheorem{proposition}{Proposition}
\newtheorem{definition}{Definition}
\newtheorem{corollary}{Corollary}

\newtheorem{assumption}{}

\newtheorem{assumptionS}{}

\newcommand{\crl}[1]{\left\{#1 \right\}}
\newcommand{\round}[1]{\left (#1 \right)}
\newcommand{\sqrd}[1]{\left[#1 \right]}

\newcommand{\Liminf}{\mathrm{Liminf}}

\usepackage{calrsfs}
\newcommand{\Ell}{\mathcal{L}}

\makeatother

\begin{document}

\title{Convergence of a class of gradient-free optimisation schemes
when the objective function is noisy, irregular, or both}

\date{\today}

\author{\thanks{c.andrieu@bristol.ac.uk} Christophe Andrieu, \thanks{nicolas.chopin@ensae.fr} Nicolas Chopin, \thanks{ettore.fincato@bristol.ac.uk} Ettore Fincato, \thanks{mathieu.gerber@bristol.ac.uk} Mathieu Gerber\\
        \small $^{*,\ddagger,\S}$ School of Mathematics, University of Bristol, UK\\
        \small $^{\dagger}$ ENSAE, Institut Polytechnique de Paris, France\\
}

\maketitle

\begin{abstract}

We investigate the convergence properties of a class of iterative algorithms 
designed to minimize a potentially non-smooth and noisy
objective function, which may be algebraically intractable and whose values may be obtained as the output of a black box. The algorithms considered can be cast under the umbrella of a generalised gradient descent
recursion, where the gradient is that of a smooth approximation of the objective function. The framework we develop includes as special cases model-based and mollification methods, two classical approaches to zero-th order optimisation.
The convergence results are obtained under very weak assumptions on the regularity of the objective function and involve a trade-off between the degree of smoothing and size of the steps taken in the parameter updates. As expected, additional assumptions are required in the stochastic case. We illustrate the relevance of these algorithms and our convergence results through a challenging classification example from machine learning.


\end{abstract}


\section{Introduction}

We are interested in the convergence properties of iterative algorithms designed to minimise an objective function $l:\R^d\to \R$ in the
following, non-exclusive, scenarios: (a) $l$ is not smooth (i.e.~non-differentiable or even
discontinuous); or (b) $l$ is the output of a black box; or (c) only  noisy evaluations of $l$ can be obtained. 
Such scenarios preclude the use of standard gradient descent methods: in  (a) the
gradient of $l$ does not exit, in  (b) little is known about $l$ and even when the gradient exists it is inaccessible, and in  (c) the function $l$ itself cannot be evaluated exactly. Numerous algorithms have been developed to address such situations, including, among many others, simulated annealing \citep{doi:10.1287/opre.18.6.1225,Khachaturyan:a19748},  genetic algorithms \citep{zhigljavsky2008stochastic}, function smoothing techniques \citep{ermolievnorkin2003} and model-based, or variational, search methods, \citep{article,ghosh2025variationallearningfindsflatter}. In this contribution we develop novel convergence theory covering the latter two approaches simultaneously.

More specifically, we assume the function $l$ to be of the form
$l(\theta)=\E[\ell(\theta, U)]$, where $U$ is a random variable whose distribution $\mathbb{P}$ is independent of $\theta$, and we let $\ell\colon\R^d\times\mathsf{U}\rightarrow\R$ be the noisy evaluation of $l$ one may obtain for any given input $\theta \in \mathbb{R}^d$; 
the noiseless scenario is recovered by taking $\ell(\theta, U) = l(\theta)$. For
  any $\gamma>0$ a smooth approximation $\Ell_\gamma \colon \R^d\times \mathsf{U}\rightarrow \R$ of $\ell$ can be defined as follows,
\begin{equation}\label{eq:smooth_approx}
  \mathcal{L}_\gamma(\theta, u)  = \psi^{-1}\left( 
  \int \psi\left( \ell(x, u) \right)
  \phi_{d,\gamma}(x - \theta) dx\right) \,,
\end{equation}
where $\phi_{d,\gamma}(x)$ is the probability density function of a $\mathscr{N}_d(0, \gamma I)$ distribution, and $\psi$ a diffeomorphism. The scalar parameter $\gamma>0$ determines both how far  $\mathcal{L}_\gamma(\theta, u)$ is from
$\ell(\theta, u)$ (the smaller, the closer), and how smooth $\mathcal{L}_\gamma(\theta, u)$ is (the larger, the smoother). This smooth approximation of  $\ell$ is differentiable under fairly mild conditions and,
in this work, for a given   
$\theta_0 \in \mathbb{R}^d$  we consider  the generalised gradient descent recursion: 
\begin{equation}\label{eq:theta_seq}
\theta_{n+1} = \theta_n - \beta_n \nabla \mathcal{L}_{\gamma_n}(\theta_n,
U_{n+1}),\quad n \geq 1,
\end{equation}
where the $U_n$'s are independent and identically distributed from $\mathbb{P}$, 
and where $(\gamma_n)_{n \geq 1}$ and $(\beta_n)_{n\geq 1}$ are sequences of positive numbers. 

In Section~\ref{subsec:mainresults} we present general results regarding the convergence of   the generic  time-inhomoge\-neous gradient descent algorithms   \eqref{eq:theta_seq}. The main results in Section~\ref{subsec:mainresults} are Theorems \ref{thm:conv-general} and \ref{thm:conv-general-det}, where the deterministic and stochastic scenarios are treated separately. Letting $L_\gamma(\theta):=\E[\Ell_\gamma(\theta,U)]$ for all $\gamma>0$ and $\theta\in\R^d$, the convergence of \eqref{eq:theta_seq} is established in the sense that $\lim_{n\rightarrow\infty}\|\nabla L_{\gamma_n}(\theta_n)\|=0$ under general conditions on the two functions $(\theta,\gamma,u)\mapsto \Ell_\gamma(\theta,u)$ and $(\theta,\gamma,u)\mapsto\nabla_\theta \Ell_\gamma(\theta,u)$, and on the two sequences $(\beta_n)_{n\geq1}$ and $(\gamma_n)_{n\geq1}$.  In Section~\ref{subsec:application-theory} (Propositions \ref{prop:2point} and \ref{prop:Bayes})  we then specialise these results to the following two popular scenarios: $\psi(x)=x$, in which case 
\begin{equation} \label{eq:def-barl-gam}
\Ell_{\gamma}(\theta,u)= 
\int_{\mathbb{R}^d} \ell(x,u)\phi_{d,\gamma}(x-\theta)\dd x \,, 
\end{equation} 
and $\psi(x)=\exp(-x)$, in which case 
\begin{equation} \label{eq:def-l-gam}
   \Ell_{\gamma}(\theta,u)=  
    -\log\Big(\int_{\R^d} e^{-\ell(x,u)}\phi_{d,\gamma}( x-\theta)\dd x\Big).
\end{equation}


The connection between these two particular instances of recursion \eqref{eq:theta_seq} with smoothing (i.e.~mollification) techniques and model-based search methods is discussed in Section \ref{subsec:links-to-the-literature}, while in Section~\ref{sec:application} we apply our theory to an intricate example from the machine learning literature.  We stress that convergence results of recursions of the type \eqref{eq:theta_seq} in scenarios \eqref{eq:def-barl-gam} and \eqref{eq:def-l-gam} are, to the best of our knowledge, scarce; see Section~\ref{subsec:links-to-the-literature} for a discussion.

The expectation involved in the expression for the gradient of $\mathcal{L}_\gamma(\theta, u)$ may not be tractable, in which case it is replaced with an  estimator based on iid samples from $\mathcal{N}_d(0,\gamma I_d)$. When   $\mathcal{L}_\gamma(\theta, u)$ is as defined in \eqref{eq:def-barl-gam} we can easily obtain an unbiased  estimator of this gradient, in which case our results  remain valid; this only requires notational change. In contrast, for  $\mathcal{L}_\gamma(\theta, u)$ as defined in \eqref{eq:def-l-gam} constructing an unbiased estimator of the gradient is not trivial  due to the presence of a ratio. Classical approaches to study such recursions rely on using the ideal dynamics \eqref{eq:theta_seq} as a reference, of which the algorithm effectively implemented is a perturbation; our result is therefore the first step in such an analysis. Alternatively, it is possible to modify the algorithm to employ debiasing techniques for self-normalised importance sampling \citep{10.3150/15-BEJ785,cardoso2022brsnisbiasreducedselfnormalized}. In this manuscript we do not attempt to compare performance of the two approaches, which is left for future work.



\section{Main results and discussion} 

\subsection{General inhomogeneous gradient descent}\label{subsec:mainresults}

This subsection presents results on the convergence of the inhomogeneous
gradient descent~\eqref{eq:theta_seq}, without making assumptions on the exact expression 
for $\Ell_\gamma(\theta, u)$, e.g., the choice of bijection $\psi$ in~\eqref{eq:smooth_approx}. 
Specifically, we provide technical, intermediate results which we will leverage in
the next subsection  to establish the convergence of~\eqref{eq:theta_seq} for the two maps
$\psi$ of interest. 

Recall  that $L_\gamma(\theta)=\E\left[\Ell_\gamma(\theta, U)\right]$ for any $\gamma > 0$. 
We consider the following assumptions on these functions. 


\begin{assumptionS}\label{assume:Gen1}
There exist finite constants $c\in\R$ and $\bar{\gamma}_1>0$  such that $L_\gamma(\theta)\geq c$ for all $\theta\in\R^d$ and all $\gamma\in (0,\bar{\gamma}_1]$.
\end{assumptionS}

\begin{assumptionS}\label{assume:Gen2}
There exist  a function $C\colon \mathsf{U} \rightarrow [1,\infty)$, such that $\mathbb{E}[C(U)^\eta]<\infty$ for some $\eta\in [2,\infty)$, and   constants  $\alpha \in [0,1]$ and $\bar{\gamma}_2\in(0,\infty]$ such that, for all $0<\tilde{\gamma}\leq \gamma\leq \bar{\gamma}_2$, all  $\theta,\theta'\in\R^d$ and all $u\in\mathsf{U}$, the following conditions hold:
   \begin{enumerate}
    \item $\|\nabla \Ell_\gamma(\theta,u)\|\leq \gamma^{\frac{\alpha-1}{2}}C(u)$,
      \item $\|\nabla \Ell_\gamma(\theta,u)-\nabla \Ell_\gamma(\theta',u)\|\leq \gamma^{-1+\frac{\alpha}{2}}\|\theta-\theta'\| C(u)$\,,
\item $\big|\Ell_{\tilde{\gamma}}(\theta,u)-\Ell_\gamma(\theta,u)\big|\leq (\gamma/\tilde{\gamma})^{\frac{d}{2}} \frac{\gamma-\tilde{\gamma}}{\tilde{\gamma}}C(u)$,
\item $\big\|\nabla \Ell_{\tilde{\gamma}}(\theta,u)-\nabla \Ell_\gamma(\theta,u)\big\|\leq \gamma^{1/2}(\gamma/\tilde{\gamma})^{\frac{d}{2}} \frac{\gamma-\tilde{\gamma}}{\tilde{\gamma}^{2}}C(u)$.
\end{enumerate}
\end{assumptionS}

\noindent In the deterministic setup, $C(u)$ is a constant and $u$ should be ignored throughout. In Section \ref{subsec:conv-general} we show the   two results presented below.
\begin{theorem}\label{thm:conv-general}
 Assume that \ref{assume:Gen1}-\ref{assume:Gen2} hold  and let $(\theta_n)_{n\geq 1}$ be as defined in \eqref{eq:theta_seq}, where $\beta_n=c_\beta n^{-\iota}$ and $\gamma_n=c_\gamma n^{-\kappa}$ for all $n\geq 1$ and for some constants $(c_\beta,c_\gamma)\in (0,\infty)^2$ and $(\iota,\kappa)\in(0,1]^2$. Let    $\alpha\in[0,1]$ and $\eta\in[2,\infty)$ be as in \ref{assume:Gen2}. Then,
\begin{enumerate}
\item if   $\kappa(2-3\alpha/2)<\iota$ we have $\liminf_{n\rightarrow\infty}\|\nabla L_{\gamma_n}(\theta_n)\|=0$, $\P-a.s$,
\item if in addition   $\min\{1-\kappa/2, \iota-\kappa(3/2-\alpha)\} >1/\eta$ then  $\lim_{n\rightarrow\infty}\|\nabla L_{\gamma_n}(\theta_n)\|=0$, $\P-a.s$.
\end{enumerate} 
\end{theorem}
 
 The deterministic scenario requires weaker assumptions on $(\beta_n)_{n\geq1}$, $(\gamma_n)_{n\geq 1}$ and no moment condition.

\begin{theorem}\label{thm:conv-general-det}
 Assume that $\ell(\theta,u)=l(\theta)$ for all $(\theta,u)\in \R^d\times \mathsf{U}$ and that \ref{assume:Gen1}-\ref{assume:Gen2} hold,  and let $(\theta_n)_{n\geq 1}$ be as defined in \eqref{eq:theta_seq}, where $\beta_n=c_\beta n^{-\iota}$ and $\gamma_n=c_\gamma n^{-\kappa}$ for all $n\geq 1$ and for some constants $(c_\beta,c_\gamma)\in (0,\infty)^2$ and $(\iota,\kappa)\in(0,1]^2$. Let $\alpha\in[0,1]$ and $\eta\in[2,\infty)$ be as in \ref{assume:Gen2}. Then, there exists a constant $c_\star>0$ such that 
\begin{enumerate}
\item if $\kappa(1-\alpha/2)<\iota$  or we have both $\kappa(1-\alpha/2)=\iota$ and $c_\beta c_\gamma^{\alpha/2-1}< c_\star$ then $\liminf_{n\rightarrow\infty}\|\nabla L_{\gamma_n}(\theta_n)\|=0$,
\item if in addition $\kappa(3/2-\alpha)<\iota$  then $\lim_{n\rightarrow\infty}\|\nabla L_{\gamma_n}(\theta_n)\|=0$.
\end{enumerate}  
\end{theorem}

\begin{remark}
For the scenario where $\psi(x)=\exp(-x)$   the conclusion of Theorem \ref{thm:conv-general-det} holds with $c_\star=2$.  
\end{remark}

The next assumption is specific to the interpretation of the provided convergence results. The definition of epi-convergence is given in Appendix \ref{app:epi_conv_lsc}.
\begin{assumptionS}
    \label{assume:Gen3}
For any sequence $(\gamma_n)_{n\geq 1}$ in $\R_+$ such that $\lim_n\gamma_n =0$, the sequence $(L_{\gamma_n})_{n\geq 1}$ epi-converges  as $n\to \infty$ to the function  $l$ (see Appendix   \ref{app:epi_conv_lsc} for a definition of epi-convergence).  
\end{assumptionS}

Under Assumption \ref{assume:Gen3}, and some minimal regularity conditions on $l$, an interpretation of the convergence results of Theorems \ref{thm:conv-general} and \ref{thm:conv-general-det} can be given in terms of the following characterisation of local minima under epi-convergence (Theorem \ref{thm:loc_argmin_result}), which is an extension of a result of \citet{Ermoliev1995} taken from  \citet[Theorem 3.2]{andrieu2024gradientfreeoptimizationintegration}.

\begin{theorem}
    \label{thm:loc_argmin_result}
Assume that \ref{assume:Gen3} holds and that the function $l$ is locally integrable, lower bounded and lower semi-continuous. Then, for any $\theta_*\in \operatorname{loc-argmin}  l$, there exists a sequence $(\theta_n)_{n\geq 1}$ such that $\lim_{n\rightarrow\infty}\theta_n= \theta_{*}$ and $\lim_{n\rightarrow\infty}\|\nabla L_{\gamma_n}(\theta_n)\|=0$.
\end{theorem}

When combined with Theorem \ref{thm:loc_argmin_result},  
Theorems \ref{thm:conv-general} and  \ref{thm:conv-general-det} constitute a tool to identify local minima candidates: if the sequence
$(\theta_n)_{n\geq 1}$ defined in \eqref{eq:theta_seq}
converges to some $\theta$, then $\theta$ is a candidate local minimum of
$l$. This also leads to a number of consequences and stronger results, when more is
known on the objective function. For instance,  if $\ell(\theta,u)=l(\theta)$ for all $(\theta,u)\in \R^d\times \mathsf{U}$, the function  $\theta \mapsto l(\theta)$ is convex and
the function  $\theta\mapsto L_{\gamma}(\theta)$ is also convex for any $\gamma>0$, by using Attouch's theorem  \citep[Theorem 3.66]{attouch1984variational} one can easily show from Theorems \ref{thm:conv-general} and  \ref{thm:conv-general-det} that the sequence $(\theta_n)_{n\geq 1}$ defined in \eqref{eq:theta_seq} converges to the minimiser of $l$.


\subsection{Application to function smoothing and model-based optimization} \label{subsec:application-theory}

In this subsection we show that Theorems \ref{thm:conv-general}, \ref{thm:conv-general-det} and \ref{thm:loc_argmin_result} hold for the particular smoothed approximations of the function $\ell$ considered in $\eqref{eq:def-barl-gam}$ and \eqref{eq:def-l-gam}, under the following assumptions:

\begin{assumption}\label{assume:lower-bound}
$\inf_{\theta\in\R^d}\E[\ell(\theta,U)]>-\infty$.
 \end{assumption}
 
\begin{assumption}\label{assume:holder}
There exist constants  $\alpha\in[0,1]$, $\eta\in[2,\infty)$ and $\beta\in [\alpha,\infty)$  and a function $J:\mathsf{U}\to \R_+$, such that $\E[J(u)^\eta]<\infty$ and such that
\begin{align}\label{eq:Holder}
|\ell(\theta,u)-\ell(\theta',u)|\leq J(u)\big(\|\theta-\theta'\|^\alpha+\|\theta-\theta'\|^\beta\big),\quad\forall \theta,\theta'\in\R^d,\quad\forall u\in\mathsf{U}.
\end{align}
\end{assumption}

\begin{assumption}\label{assume:holder2}
There exist constants  $\alpha\in[0,1]$,  $\beta\in [\alpha,2]$ and $\upsilon\in(0,1]$ and and a function $J:\mathsf{U}\to \R_+$, such that \eqref{eq:Holder} holds and such that
\begin{align*}
\E\bigg[\int_{\R } e^{\upsilon J(U)(1+|z|^\beta)}\phi_{1,1}(z)\dd z\bigg]<\infty,\quad U\sim\P. 
\end{align*} 
\end{assumption}

Under Assumptions \ref{assume:lower-bound}-\ref{assume:holder2} we have the following two results, proved in Section \ref{proof_prop12}.

\begin{proposition}\label{prop:2point}
  Assume that \ref{assume:lower-bound} and \ref{assume:holder}. hold. Then, \ref{assume:Gen1}-\ref{assume:Gen2} hold  for $\psi(x)=x$, that is for $\Ell_\gamma$ defined 
  as in~\eqref{eq:def-barl-gam},  with $\alpha$ and $\eta$ as in  \ref{assume:holder}.
\end{proposition}

 \begin{proposition}\label{prop:Bayes}
 Assume that \ref{assume:lower-bound} and \ref{assume:holder2}   hold. Then, \ref{assume:Gen1}-\ref{assume:Gen2} hold  for both $\psi(x)=x$, and $\psi(x)=\exp(-x)$, or, in 
 other words, for $\Ell_\gamma$ defined either as~\eqref{eq:def-barl-gam} or~\eqref{eq:def-l-gam}, 
 and with $\alpha$ as in \ref{assume:holder2} and for any $\eta>0$.
 \end{proposition}

Both results rely on \eqref{eq:Holder} where the term depending on $\alpha$ controls the local regularity of $\theta \rightarrow \ell(\theta,u)$. Of particular interest is the scenario $\alpha = 0$ which allows one to consider bounded discontinuities. The term depending on $\beta$ controls large variations, and in particular the behaviour at infinity of this function. The main difference between the two families of smoothed approximations is the moment condition on $J(u)$, which is more stringent when $\psi(x)=\exp(-x)$ due to the presence of the exponential and the requirement to control a ratio.

We move to the characterisation of local minima under epi-convergence, showing that the assumptions of Theorem \ref{thm:loc_argmin_result}  hold  for the smooth approximations of $\ell$ defined in \eqref{eq:def-barl-gam} and \eqref{eq:def-l-gam}.
\begin{proposition}
    \label{prop:characterisation_specific}
    Assume \ref{assume:lower-bound}, \ref{assume:holder} and that $\theta\mapsto \ell(\theta,u)$ is lower-semicontinuous for $\P$-a.e. $u\in\mathsf{U}$. Assume further that for all $  \theta\in \R^d$ there exists a sequence $(\theta_n)_{n\geq 1}$ such that 
    $\lim_{n\rightarrow\infty}\theta_n= \theta$, such that $l$ is continuous at $\theta_n$ for every $n\geq 1$, and  such that $\lim_{n\rightarrow\infty} l(\theta_n)=l(\theta)$. Then \ref{assume:Gen3} holds for  $\Ell_\gamma$ defined either as~\eqref{eq:def-barl-gam} or~\eqref{eq:def-l-gam},
    and the function $l$ is a locally integrable, lower bounded and lower semi-continuous.
\end{proposition}
 
We finally remark that if $\ell(\theta,u)=l(\theta)$ for all $(\theta,u)\in \R^d\times \mathsf{U}$ and $l$ is convex, then it is easy to see that $\Ell_{\gamma}$, is also convex for any $\gamma>0$, 
(whether~\eqref{eq:def-barl-gam} or~\eqref{eq:def-l-gam} hold), 
leading to stronger convergence results provided by the previously mentioned Attouch theorems.

 \subsection{Links to the literature} \label{subsec:links-to-the-literature}

\cite{gupal_norkin} also consider the optimisation of a discontinuous function using a recursion similar to \eqref{eq:theta_seq} with $\Ell_\gamma$ of the form \eqref{eq:def-barl-gam}, but where the smoothing kernel is uniform, requiring the use of two successive convolutions to ensure differentiability. The nature of their theoretical result is similar to ours, but exploits boundedness of the kernel support and assumes the sequence $(\theta_n)_{n \geq 1}$ to be bounded, which we do not require. In addition, we show that the use of a Gaussian kernel leads to more favourable dependence on the smoothing parameter of the gradient, therefore relaxing assumptions on step-sizes. Gaussian smoothing is mentioned in \citet{Ermoliev1995} in a remark while \citet{nesterov2017random} seems to be the first thorough study of a version of the algorithm. More specifically, in the scenario where $(\gamma_n)_{n\geq 1}$ is constant, they obtain complexity bounds under the assumption of convexity of $l$ \citet[][Sections 4 and 5]{nesterov2017random} or L-smoothness or global Lipschitz continuity \citet[][Section 7]{nesterov2017random}, therefore not covering the discontinuous setup.

\cite{starnes2023gaussian} also consider the algorithm discussed above, which they call GSMoothGD (Gaussian Smoothing Gradient Descent), with constant step-sizes and potentially vanishing smoothing. Under the assumption of L-smoothness on $\theta\mapsto l(\theta)$ they establish convergence estimates of the function values along the iterations in the convex scenario \citep[][Theorem 3.1]{starnes2023gaussian} and bounds on the minimum gradient norm encountered in the non-convex scenario \citep[][Theorem 3.3]{starnes2023gaussian}, akin to a $\liminf$ result. Our results complement theirs by considering rougher functions, including potentially discontinuous ones, and utilising vanishing step-sizes, both in the deterministic and stochastic scenarios. Not surprisingly, the results we obtain are weaker given the weaker assumptions considered and in line with the corresponding literature \citep{gupal_norkin,Ermoliev1995}. In \cite{starnes2024anisotropicgaussiansmoothinggradientbased} the authors consider extensions to the stochastic scenario and where the covariance matrix of the normal kernel is adapted. In both contributions, the authors also provide an in-depth literature review, discussing links with natural evolution strategy algorithms.

In our earlier contribution \citep{andrieu2024gradientfreeoptimizationintegration}, we established convergence of \eqref{eq:theta_seq},  in the situation where 
(a) $\Ell_\gamma$ is as in \eqref{eq:def-l-gam},
(b) in the deterministic scenario and 
(c) for $(\beta_n)_{n\geq 1}=(\gamma_n)_{n\geq 1}$ only. Here, thanks to a particular effort to obtain good estimates of the quantities in \ref{assume:Gen2} (see Lemma \ref{lemma:estimates-laplace}), we lift these three restrictions, leading to broader applicability. We further establish epiconvergence of $\Ell_\gamma(\theta,u)$ to $\ell(\theta,u)$ as $\gamma \rightarrow 0$, required for the interpretation of accumulation points \citep{Ermoliev1995} in the discontinuous scenario. 

Interestingly, this work also establishes a link to the model-based search literature \citep{article,rubinstein2004cross}, where convergence results seem scarce or specific \citep[e.g.,][]{COSTA2007573}, by reinterpreting the following procedure as a recursion of the form \eqref{eq:theta_seq}. With here $\pi_{\theta,\gamma}(x):=\phi_{d,\gamma}(x-\theta)$ for $(\theta,\gamma) \in \R^d\times \R_+$ and $(\gamma_n)_{n\geq 1}$ a vanishing sequence of step sizes, the algorithm consists of constructing a sequence of probability densities $(\pi_{\theta_n,\gamma_n})_{n\geq 1}$  concentrating on local minima of $l$. Specifically, for a sequence of random variables $(U_n)_{n\geq1}$ as above, the recursive algorithm proceeds as follows. At iteration $n+1$, (a) apply a pseudo Bayes' rule
$\tilde{\pi}_{n+1}(x)\propto \pi_{\theta_n,\gamma_n}(x)\exp\left\{-\ell(x,U_{n+1})\right\}$ (b) project $\tilde{\pi}_{n+1}$ back onto the Gaussian family, i.e.~obtain $\pi_{\theta_{n+1},\gamma_{n+1}}$ with $\theta_{n+1}\in \argmin_{\theta\in \Theta} \mathcal{\mathrm{KL}}(\tilde{\pi}_{n+1}, \pi_{\theta,\gamma_{n}})$.
 It can be shown that the resulting sequence $(\theta_n)_{n \geq 1}$ can be written as \eqref{eq:theta_seq} in the situation where $\Ell_\gamma$ is as in \eqref{eq:def-l-gam}. In the deterministic scenario, this algorithm can also be interpreted as a greedy coordinate descent algorithm applied to the functional
$\Phi_1(\nu,\theta;\gamma)=\int l(x)\nu(x){\rm d}x+{\rm KL}(\nu,\pi_{\theta,\gamma})$; the stochastic scenario is similar. Indeed, one can check (see Appendix~\ref{app:coordinate-descent}) that each iteration amounts to alternating Bayes' rule and the minimisation of
\begin{equation*}
\theta \mapsto l_\gamma(\theta) \coloneq -\log\left(\int e^{-l(x)}
\gamma^{-d/2}\phi_{d,\gamma}\left(\frac{x-\theta}{\sqrt{\gamma}}\right)\mathrm{d} x\right) \,,
\end{equation*}
which, as discussed below, can be achieved by moment matching. Noting that
\begin{equation}
\nabla l_\gamma(\theta) \coloneq -\gamma^{-1}[\E_{\tilde{\pi}_n}(X)-\theta],
\end{equation}
one can replace local exact minimisation of $l_\gamma(\theta)$ with a gradient descent step leading, for a sequence of step-sizes $(\beta_n)_{n\geq 1}$, to the recursion \eqref{eq:theta_seq} taking the form
\begin{equation}\label{eq:gradient_descent_with_moment_match}
    \theta_{n+1} = \theta_n -\beta_n \gamma_n^{-1}[\theta_n-\E_{\tilde{\pi}_n}(X)] \,.
\end{equation}
We notice that when $\beta_n = \gamma_n$ we recover moment matching $\theta_{n+1}=\E_{\tilde{\pi}_n}(X)$, the particular scenario covered in \cite{andrieu2024gradientfreeoptimizationintegration}. Decoupling $(\beta_n)_{n\geq 1}$ and $(\gamma_n)_{n\geq 1}$ however allows for better control of the irregularities induced by the discontinuities of the objective function.

\section{Example: A discontinuous and noisy optimization problem} \label{sec:application}

\subsection{Set-up}

We illustrate the relevance of our theory on a classification example based on \citet{ustats}. We let $\{(z_{i},y_{i})\in
\mathbb{R}^{p}\times\{-1,1\}\colon i=1,\ldots,n_\mathrm{data}\}$ be a training dataset (with $p\geq 2$), assumed to
arise from a probability distribution $\mathbb{P}$, and we  wish to construct a
score function $s\colon \mathbb{R}^{p} \to \mathbb{R}$ such that the empirical version of the probability
\begin{align}\label{eq:auc}
\mathbb{P}\big(\left[ s(Z)-s(Z')\right] (Y-Y')<0\big)\,,\quad (Z,Y),(Z',Y')\iid  \mathbb{P}
\end{align}
 is as small as possible. The quantity \eqref{eq:auc} is often called the area under curve (AUC) risk function, and  this criterion has the advantage to be less sensitive to class imbalance
than other more standard classification criteria.  

For any $\vartheta\in\R^p$ we let $s_\vartheta(z)=\vartheta^\top z$ for all $z\in\R^p$, and below we focus  on the scenario where the function $s$ must be chosen from the set $\{s_\vartheta,\,\vartheta\in \R^p\}$. Without loss of generality,  we can assume that the observations in the training sets are labeled so that   
\begin{align*}
y_i=
\begin{cases}
1  &i\in\{1,\ldots,n_{+}\}\\
-1 &i\in\{n_{+}+1,\ldots,n_{\mathrm{data}}\}
\end{cases},\quad  n_{+}:=\sum_{i=1}^{n_{{\rm data}}}\mathbf{1}(y_{i}=1).
\end{align*}  
Letting $I=\{(i,j)\colon i\in\{1,\ldots,n_{+}\},j\in\{n_{+}+1,\ldots,n_{\mathrm{data}}\}\}$, the parameter $\vartheta$ is then chosen by minimizing the empirical risk  $E:\R^p\mapsto [0,\infty)$, defined by
\begin{equation} 
\begin{split}
E(\vartheta)&=\frac{\sum_{(i,j)\in I}\mathbf{1}\{s_{\vartheta}(z_{i})<s_{\vartheta}(z_{j})\}}{n_\mathrm{data}(n_\mathrm{data}-1)}=\frac{\sum_{(i,j)\in I}\mathbf{1}\{\vartheta^{\top}(z_{i}-z_{j})<0\}}{n_\mathrm{data}(n_\mathrm{data}-1)},\quad \quad  \vartheta\in\R^p.\label{eq:barl-AUC}
\end{split}
\end{equation}

\subsection{Definition of the functions $l(\theta)$ and $\ell(\theta,u)$ }

While minimising  the empirical risk $E(\vartheta)$ is a classical problem in machine learning, two issues arise when this procedure is used for choosing  a score function $s\in\{s_\vartheta,\,\vartheta\in\R^p\}$. The first problem is that $\vartheta=0$ is a global minimum of $E(\cdot)$, which
is pathological, and the second problem is that the mapping  $\vartheta \mapsto E(\vartheta)$ is rescaling invariant.

To address these two  problems,  we can exploit the fact that for $\vartheta\neq0$ we have $E(\vartheta)=E(\vartheta/\|\vartheta\|)$ to restrict our attention to minimizers of the function $E$ that lie on  $ \mathcal{S}_p$,  the unit hypersphere in $\R^p$. To do so, we let $\sigma\colon \mathcal{S}_p \setminus \{\mathbf{e}_p\} \mapsto\mathbb{R}^{p-1}$ be the stereographic projection of centre $\mathbf{e}_{p}:=(0,\ldots,0,1)\in\mathbb{R}^{p}$, defined by 
\begin{equation}
    \sigma(\vartheta)=\left(  \frac{\vartheta_1}{1-\vartheta_p}, \frac{\vartheta_2}{1-\vartheta_p},\ldots,  \frac{\vartheta_{p-1}}{1-\vartheta_{p}}\right),\quad (\vartheta_1,\dots,\vartheta_p)\in  \mathcal{S}_p.
\end{equation}
It can be easily shown that the mapping $\sigma$ is bijective, with inverse mapping  $\sigma^{-1}\colon\mathbb{R}^{p-1}\rightarrow\mathcal{S}_p$, which is  Lipschitz continuous and defined by
\[
\sigma^{-1}(\theta ) = 
\left( \frac{2\theta_{1}}{\|\theta\|^{2}+1},\ldots,\frac{2\theta_{d-1}}{\|\theta\|^{2}+1},\frac{\|\theta\|^{2}-1}{\|\theta\|^{2}+1}\right),\quad(\theta_1,\dots,\theta_{p-1})\in \R^{p-1}.
\] 
With this notation in place,  our classification problem can be   reformulated as choosing an element in the set $\{s_{\sigma(\theta)},\,\theta\in\R^{p-1}\}$ by minimizing the function $l:\R^{p-1}\rightarrow [0,\infty)$ defined by 
\begin{align*} 
l(\theta)= E\circ\sigma^{-1}(\theta),\quad \theta\in\R^{p-1}.
\end{align*}

\begin{remark}
This definition of the function $l$ implicitly assumes that the point $\mathbf{e}_{p}=(0,\ldots,0,1)\in\mathbb{R}^{p}$ is not a minimizer of the empirical risk $E(\cdot)$, since there exists no $\theta\in \R^{p-1}$ such that $\sigma^{-1}(\theta)=\mathbf{e}_{p}$. 
\end{remark}

To reduce computations one may replace the sum appearing in \eqref{eq:barl-AUC} by an
unbiased estimate, following the standard mini-batch approach popular in machine
learning.  To this aim let $n_{\mathrm{batch}}\in\{1,\dots,\sharp I\}$  denote the batch size and, for any element  $u=\big( (i_{1},j_{1}),\dots,(i_{n_{\mathrm{batch}}},j_{n_{\mathrm{batch}}})\big)$ of  $\mathsf{U}:=I^{n_{\mathrm{batch}}}$,  let
\begin{equation}
\mathcal{E}(\vartheta,u) = \frac{2n_{+}(n_{\mathrm{data}}-n_+)}{n_\mathrm{data}(n_\mathrm{data}-1)n_{\mathrm{batch}}}\sum_{k=1}^{n_{\mathrm{batch}}}\mathbf{1}\{s_{\vartheta}(z_{i_{k}})<s_{\vartheta}(z_{j_{k}})\},\quad\forall \vartheta\in \R^p.\label{eq:def_auc_minibatch}
\end{equation}

With this notation in place, we let $\ell(\theta,u)=\mathcal{E}(\sigma^{-1}(\theta), u)$ for all pairs $(\theta,u)\in\R^{p-1}\times  \mathsf{U}$. Then, letting  $\mathbb{P}$ denote the uniform distribution on $\mathsf{U}$, we can easily check that for all $\theta\in\R^{p-1}$ we have $\E[\ell(\theta,U)]=l(\theta)$ if $U\sim \mathbb{P}$. Observe that the two functions $\theta\mapsto l(\theta)$ and $\theta\mapsto \ell(\theta,u)$ are discontinuous.


Since each term in the sum appearing in \eqref{eq:def_auc_minibatch} is equal to zero or one,  it readily follows that \ref{assume:lower-bound} holds while \eqref{eq:Holder} is satisfied for some bounded function $J:\mathsf{U}\to \R_+$ and with $\alpha=0$. In this context,   condition  \ref{assume:holder} holds for any $\eta\in [2,\infty)$ and condition \ref{assume:holder2} holds for any $\beta\in[0,2]$ and $\upsilon\in(0,1]$. The following proposition, proved in Section \ref{p-propExample}, establishes   that the functions $l$ and $\ell$ defined above also satisfy the assumptions of Proposition \ref{prop:characterisation_specific}.

\begin{proposition}\label{prop:lsc-unconstrained-example}
For any $\vartheta\in \R^d$ let $D_{\vartheta}=\left\{ (i,j)\in I\colon\vartheta^{\top}(z_{i}-z_{j})=0\right\}$ and assume that the following condition holds
\begin{align*}
\vartheta\neq 0\implies (z_{i_{1}}-z_{j_{1}})^{\top}(z_{i_{2}}-z_{j_{2}})\geq 0,\quad \forall (i_{1},j_{1}),(i_{2},j_{2})\in D_{\vartheta}.
\end{align*}
Then, the functions $l$ and $\ell$ defined above satisfy the assumptions of Proposition~\ref{prop:characterisation_specific}.
\end{proposition}

\begin{remark}
The condition imposed on the observations may appear strong. However we remark that it holds with probability
one if, for example,  the observations $z_{1},\ldots,z_{n_{{\rm data}}}$ are independent realizations of a distribution with positive density
w.r.t.~the Lebesgue measure. It is indeed highly improbable that more than one of the differences $\{z_i-z_j, (i,j)\in I\}$ belong to a hyperplane $\vartheta^\top z = 0$ for some $\vartheta \in\mathbb{R}^{d}$. In other words for any $\vartheta\in\mathbb{R}^{d}\setminus\{0\}$ we have $\sharp D_{\vartheta}\in\{0,1\}$. 
\end{remark}

\subsection{Numerical experiments}

\newcommand{\ndat}{n_\text{data}}
\newcommand{\nbatch}{n_\text{batch}}
\newcommand{\ind}{\mathbf{1}}  
\newcommand{\Cp}{\mathcal{C}^+} 
\newcommand{\Cn}{\mathcal{C}^-} 

We consider the Fashion-MNIST dataset \citep{fashion},
which consists of $\ndat=7\times 10^4$ grayscale images,
of size $p=28\times 28 = 784$, of fashion items organised in ten categories
(t-shirt/top, trousers, etc.). We set $y_i=1$ if item $i$ falls in the first
category (t-shirt/top), $y_i=-1$ otherwise, leading to class unbalance 
($n_{+} = 7\,000$, $n_{-} = 63\ 000$). Recall that one of the appeals of the AUC criterion 
is that it is less sensitive to class unbalance. 

For this dataset, even a single evaluation of the AUC criterion takes several
seconds on a recent computer, making the optimisation of this function
particularly challenging. 
To address this issue, we apply two standard machine learning tricks: 
(a) we split the data into a training set (90\% randomly selected datapoints) and a test 
set (the remaining 10\%).  Our optimiser is run on the former, and evaluated on the latter.
(b)  we use the mini-batch strategy described in the previous sub-section, 
with $\nbatch = 10^3$.
  


We set $\gamma_n = \gamma_1 / n^\kappa$, with $\gamma_1 = 0.2$, $\kappa=0.2$,
and $\beta_n = \beta_1 / n^\iota$, with $\beta_0=0.2$, $\iota=0.5$. 
(Thus, the conditions of Theorem~\ref{thm:conv-general} hold, since we can take $\alpha=0$.)
We focus on the map $\psi(x)=\exp(x)$, and thus the gradient descent recursion takes the form given 
by~\eqref{eq:gradient_descent_with_moment_match}. 
As in~\cite{andrieu2024gradientfreeoptimizationintegration}, 
we replace the expectation with an importance sampling estimate computed over a set of
$N=1024$ i.i.d. (independent and identically distributed) Gaussian variates.
In addition, we use the same rescaling strategy as in the aforementioned paper; that is, remarking that 
the minimiser of $l(\theta)$ does not change if we multiply $l$ by an arbitrary, positive constant, 
we rescale $l$ at each iteration to keep the ESS (effective sample size), or equivalently,
the variance of the normalised importance sampling weights, constant. 
Formally, our convergence results do not result when we apply this rescaling strategy, but,
in practice, rescaling seems to accelerate convergence significantly. Analysing the convergence
of the algorithm based on rescaling (and Monte Carlo estimates) is left to future work;
see next section.

We run our gradient descent algorithm 10 times. 
Figure~\ref{fig:fashion_AUC_vs_iter} shows how the AUC risk (evaluated on the
test data, at $\vartheta_n=\sigma^{-1}(\theta_n)$, the current
iterate at iteration $n$) evolves over time. 
Remarkably, convergence occurs essentially in about $10^3$ iterations, which amounts
to less than 3 epochs; an epoch  is  the number of iterations required to
access once (on average) each of the data pairs in $\mathsf{U}$.

For reference, we also include as a baseline the AUC risk of $L^2-$penalised
logistic regression (as implemented in python package scikit-learn), which is
another method to construct a linear score for predicting the class.  Of
course, this is just a sanity check, since logistic regression provides a linear score 
which is obtained by minimising a likelihood function (rather than the AUC criterion).
Interestingly, the default implementation in scikit-learn issues a warning
message  suggesting lack of convergence, which is another indication that the
considered dataset is particularly challenging.

\begin{figure}
\begin{center}
  \includegraphics[scale=0.45]{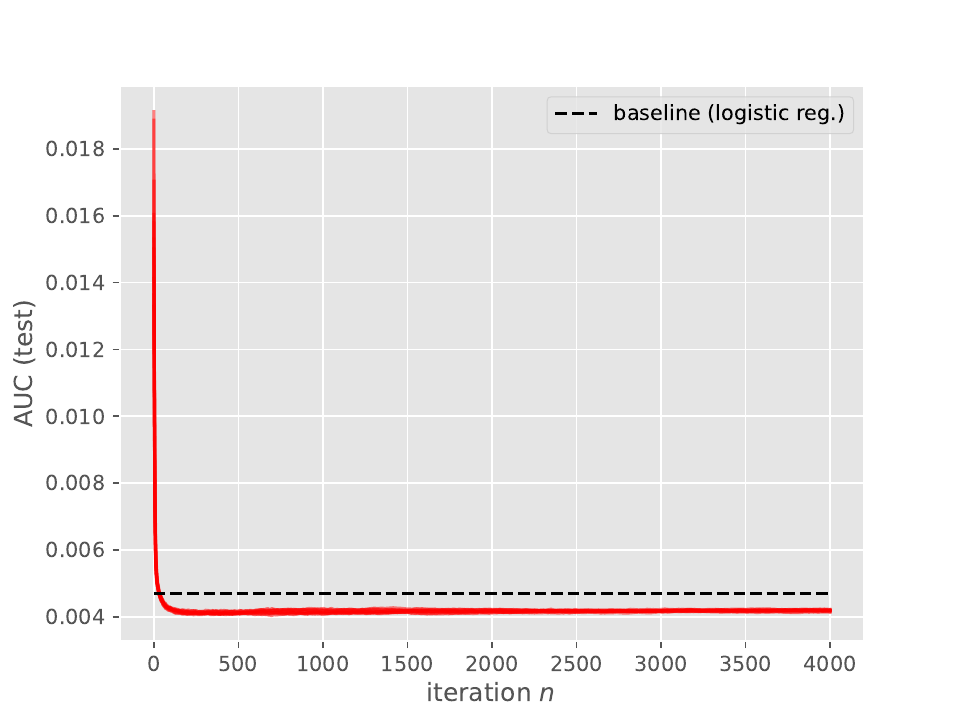}
  \includegraphics[scale=0.45]{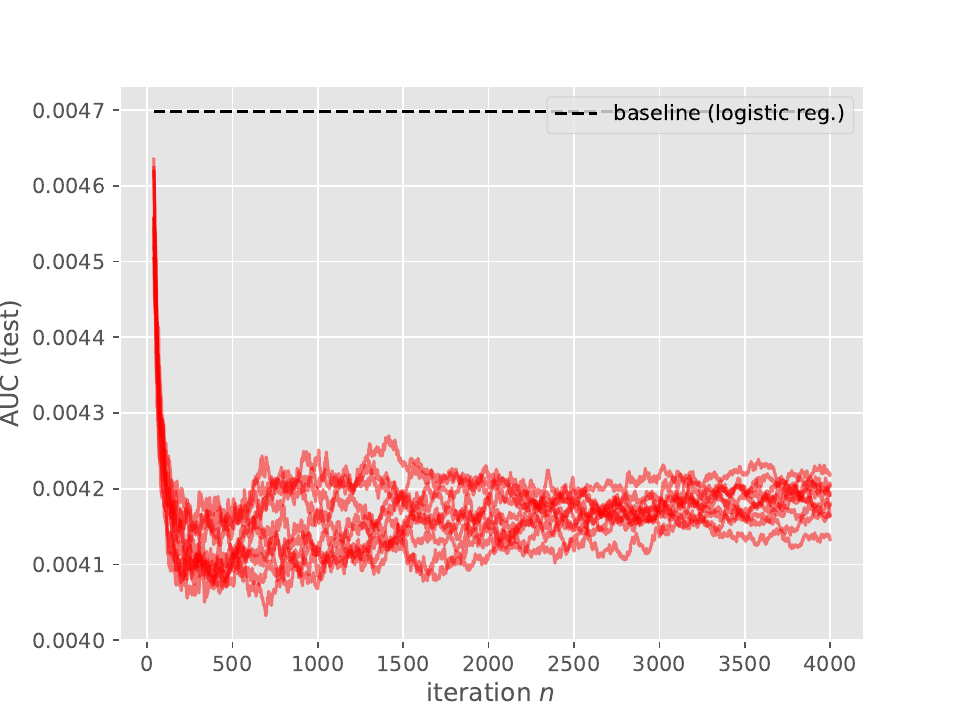}
\end{center}
\caption{AUC score (test data) of iterate vs iteration (10 independent runs).
The right panel is a zoomed-in section of the left panel, where the first 100
iterations are discarded. 
  The baseline (dashed line) is the AUC score of the logistic regression
  estimate.}
\label{fig:fashion_AUC_vs_iter}
\end{figure}

\section{Conclusion and future work}

As already discussed, this paper generalises significantly our previous results 
\citep{andrieu2024gradientfreeoptimizationintegration}, in various ways (i.e.,
noisy and noiseless scenarios, $\beta_n \neq \gamma_n$, and the map $\psi$ 
in~\eqref{eq:smooth_approx} may be either $\psi(x)=x$, or $\psi(x)=\exp(-x)$). 

Nevertheless, practical implementations of the algorithms are not yet covered,
such as the Monte Carlo variant of~\eqref{eq:gradient_descent_with_moment_match} 
(when the expectation is replaced by a biased self-normalised importance sampling estimate),  or 
the adaptive scaling strategy we use in our numerical experiment, despite our empirical observation of its beneficial impact in terms of convergence speed. In addition, there may be scenarios where one would like to use a non-Gaussian
kernel; for instance, when the domain of the objective function is a manifold
or a set of a nature different from that of $\mathbb{R}^d$.  In this vein, we have already obtained
preliminary results for smoothing kernels derived from the Wishart distribution, addressing the situation where the domain of interest is the set of positive symmetric matrices of size $d$. Another useful extension would involve relaxing our assumptions on the noise, and in particular allow for its distribution to be state dependent. 

We however leave all these extensions to future work.


\section{Proofs} \label{sec:proof}

In what follows, for any $\gamma>0$ and $(\theta,u)\in \R^d\times\mathsf{U}$ we let 
\begin{equation*}  
\bar{\ell}_\gamma(\theta,u)= 
\int_{\mathbb{R}^d} \ell(x,u)\phi_{d,\gamma}(x-\theta)\dd x,\quad\bar{l}_\gamma(\theta)=\E[\bar{\ell}_\gamma(\theta,U)]
\end{equation*} 
and 
\begin{equation*}  
   \ell_\gamma(\theta,u)=  
    -\log\Big(\int_{\R^d} e^{-\ell(x,u)}\phi_{d,\gamma}( x-\theta)\dd x\Big),\quad l_\gamma(\theta)=\E[\ell_\gamma(\theta,U)].
\end{equation*}
 
\subsection{Proof of Theorems \ref{thm:conv-general}-\ref{thm:conv-general-det}} \label{subsec:conv-general}

Theorems \ref{thm:conv-general}-\ref{thm:conv-general-det} directly  follow from   Lemmas \ref{lem:Liminf-general}-\ref{lem:Limsup-general} stated and proved in this subsection.

\subsubsection{A useful technical lemma}
\begin{lemma}\label{lemma:sequence}
Let $(a_n)_{n\geq 1}$ be a sequence in $\R$ such that $\inf_{n\geq 1}a_n>-\infty$, $(b_n)_{n\geq 1}$ be a sequence in $[0,\infty)$, and   let   $(\epsilon_n)_{n\geq 1}$ and $(\varrho_n)_{n\geq 1}$ be two bounded sequences in $[0,\infty)$ such that   $\lim_{n\rightarrow\infty}(\varrho_n/\epsilon_n)=0$ and such that $\sum_{n\geq 1}\epsilon_n=\infty$. Assume that $a_{n+1}\leq a_n-\epsilon_n b_n+\varrho_n$ for all $n\geq 1$. Then, $\liminf_{n\rightarrow\infty}b_n=0$.
\end{lemma}
\begin{proof}
We prove the result by contradiction and assume that there exists a constant $b>0$ and an $n_1\in\mathbb{N}$ such that $b_n\geq b$ for all $n\geq n_1$. We also assume that $n_1$ is sufficiently large so that $\varrho_n\leq (b/2)\epsilon_n$ for all $n\geq n_1$. Then, for all $n\geq n_1$ we have $a_{n+1}\leq a_n-\epsilon_n b+\varrho_n$ implying that, for all $n>n_1$,
\begin{align*}
a_{n}\leq a_{n_1}-b\sum_{m=n_1}^{n-1}\epsilon_m+\sum_{m=n_1}^{n-1}\varrho_m=a_{n_1}-\sum_{m=n_1}^{n-1}\epsilon_m\Big(b-\frac{\sum_{m=n_1}^{n-1}\varrho_m}{\sum_{m=n_1}^{n-1}\epsilon_m}\Big)\leq a_{n_1}-\frac{b}{2}\sum_{m=n_1}^{n-1}\epsilon_m
\end{align*}
and thus that $\lim_{n\rightarrow\infty}a_n=-\infty$. This contradicts the fact that, by assumption, $\inf_{n\geq 1} a_n>-\infty$ and the proof of the lemma is complete.
\end{proof}

\subsubsection{A descent lemma}
\begin{lemma} \label{lem:descent-general}
 Assume that \ref{assume:Gen2} holds  and let $(\theta_n)_{n\geq 1}$ be as defined in \eqref{eq:theta_seq} for some non-increasing sequence $(\gamma_n)_{n\geq 1}$ such that $\limsup_{n\rightarrow\infty}\gamma_n<\bar{\gamma}_2$, with $\bar{\gamma}_2$  as in \ref{assume:Gen2}. Then, there exist constants $\bar{C}<\infty$ and   $n_1\in\mathbb{N}$ such that, for all $n\geq n_1$ and with $\delta_n \coloneq (\gamma_n/\gamma_{n+1})^{\frac{d}{2}}\gamma_{n+1}^{-1}\big(\gamma_n-\gamma_{n+1}\big)$, we have
   \begin{align}\label{eq:des_expect}
 \E[L_{\gamma_{n+1}}(\theta_{n+1})]\leq \E[L_{\gamma_n}(\theta_n)]-\beta_n\E\big[\|\nabla L_{\gamma_n} (\theta_n)\|^2\big]+\bar{C}\Big(\beta_n^2\gamma_n^{\frac{3\alpha}{2}-2} + \delta_n \Big).
 \end{align}
In addition,  if $\ell(\theta,u)=l(\theta)$  for all $(\theta,u)\in \R^d\times \mathsf{U}$, it also holds that
 \begin{equation}\label{eq:des_det}
  \begin{split}
L_{\gamma_{n+1}}(\theta_{n+1})&\leq L_{\gamma_n}(\theta_n)-\beta_n \big(1-\bar{C}\beta_n \gamma_n^{-1+\alpha/2})\|\nabla L_{\gamma_n}(\theta_n)\|^2 +\bar{C}\delta_n,\quad\forall n\geq n_1.
\end{split}
 \end{equation}
\end{lemma}

\begin{remark}
For $\Ell_\gamma=\ell_\gamma$ we know that \eqref{eq:des_det} holds with $\bar{C}=1/2$ \citep{andrieu2024gradientfreeoptimizationintegration}. 
\end{remark}

\begin{proof}

Let $\bar{\gamma}_2>0$ and $C:\mathsf{U}\rightarrow [1,\infty)$ be as in \ref{assume:Gen2}, and let $\theta,\theta'\in\R^d$,   $u\in\mathsf{U}$ and  $\gamma\in (0,\bar{\gamma}_2]$ be arbitrary. From Taylor's theorem, 
\begin{align*}
\Ell_{\gamma}(\theta',u)\leq \Ell_{\gamma}(\theta,u)+\langle\nabla \Ell_{\gamma}(\theta,u), \theta'-\theta\rangle+\|\theta'-\theta\|\sup_{t\in [0,1]}\big| \nabla \Ell_{\gamma}(\theta,u)-\nabla \Ell_{\gamma}(\theta+t(\theta'-\theta),u)\big| 
\end{align*} 
and thus,  under    \ref{assume:Gen2} (point 2), 
\begin{equation*}
  \begin{split}
\Ell_{\gamma}(\theta',u)&\leq \Ell_{\gamma}(\theta,u)+\langle \nabla \Ell_{\gamma}(\theta,u),\theta'-\theta\rangle+\gamma^{-1+\frac{\alpha}{2}}\|\theta-\theta'\|^2 C(u)\,,
\end{split}
\end{equation*}
from which we obtain that (by taking the expectation with respect to $U$)
\begin{equation}\label{eq:l_L}
L_{\gamma}(\theta')\leq L_{\gamma}(\theta)+\langle \nabla L_{\gamma}(\theta),\theta'-\theta\rangle+\gamma^{-1+\frac{\alpha}{2}}\|\theta'-\theta\|^2 \E[C(U)].
 \end{equation}
We now let $n_1\in\mathbb{N}$ be such that $\gamma_n\leq\bar{\gamma}_2$ for all $n\geq n_1$, noting that such an $n_1$ exists under the assumptions of the lemma. Then, using  \eqref{eq:l_L} and \eqref{eq:theta_seq},  it follows that for all $n\geq n_1$ and $\mathbb{P}$-almost surely,
  \begin{equation}\label{eq:cont2}
  \begin{split}
L_{\gamma_n}(\theta_{n+1})&\leq L_{\gamma_n}(\theta_n)-\beta_n\langle\nabla L_{\gamma_n}(\theta_n),\nabla \Ell_{\gamma_n}(\theta_n,U_{n+1})\rangle+\beta_n^2\gamma_{n}^{-1+\frac{\alpha}{2}}\|\nabla \Ell_{\gamma_n}(\theta_n,U_{n+1})\|^2\,\,  \E[C(U)]\\
&\leq L_{\gamma_n}(\theta_n)-\beta_n\langle\nabla L_{\gamma_n}(\theta_n),\nabla \Ell_{\gamma_n}(\theta_n,U_{n+1})\rangle+\beta_n^2\gamma_{n}^{\frac{3\alpha}{2}-2} C(U_{n+1})^2 \times \E[C(U)] 
\end{split}
 \end{equation}
 where the second inequality holds under  \ref{assume:Gen2} (first point). Finally, using  \eqref{eq:cont2} and the third point of \ref{assume:Gen2}, we obtain that for all $n\geq n_1$
 \begin{align*}
 \E[L_{\gamma_{n+1}}(\theta_{n+1})]\leq \E[L_{\gamma_n}(\theta_n)]-\beta_n\E\big[\|\nabla L_{\gamma_n}(\theta_n)\|^2\big]+\bar{C}\Big(\beta_n^2\gamma_n^{\frac{3\alpha}{2}-2} +\delta_n\Big)
 \end{align*}
 with $\bar{C}=\E[C(U)^2] \, \E[C(U)]+\E[C(U)]<\infty$, showing the first part of the lemma. (Recall that $\E[U^\eta]<+\infty$ for some $\eta\geq 2$.)
 
To show the second part of the lemma, remark that if $\ell(\theta,u)=l(\theta)$ for all $(\theta,u)\in \R^d\times \mathsf{U}$ then the first inequality in \eqref{eq:cont2} implies that, for all $n\geq n_1$,
   \begin{equation*}
  \begin{split}
L_{\gamma_n}(\theta_{n+1})&\leq L_{\gamma_n}(\theta_n)-\beta_n\langle\nabla L_{\gamma_n}(\theta_n),\nabla L_{\gamma_n}(\theta_n)\rangle+\beta_n^2\gamma_{n}^{-1+\frac{\alpha}{2}}\|\nabla L_{\gamma_n}(\theta_n)\|^2\,\,  \E[C(U)]\\
&=L_{\gamma_n}(\theta_n)-\beta_n\big(1-\E[C(U)]\beta_n\gamma_n^{\frac{\alpha}{2}-1}\big)\|\nabla L_{\gamma_n}(\theta_n)\|^2\\
&\leq L_{\gamma_n}(\theta_n)-\beta_n\big(1-\bar{C}\beta_n\gamma_n^{\frac{\alpha}{2}-1}\big)\|\nabla L_{\gamma_n}(\theta_n)\|^2
\end{split}
 \end{equation*}
 with $\bar{C}<\infty$ as defined above. The result in the  second part of the lemma then follows from    \ref{assume:Gen2} (third part). The proof of the lemma is complete. 
\end{proof}

\subsubsection{Convergence along a subsequence}
\begin{lemma} \label{lem:Liminf-general}
 Assume that \ref{assume:Gen1}-\ref{assume:Gen2} hold  and let $(\theta_n)_{n\geq 1}$ be as defined in \eqref{eq:theta_seq}, where $\beta_n=c_\beta n^{-\iota}$ and $\gamma_n=c_\gamma n^{-\kappa}$ for all $n\geq 1$ and for some constants $(c_\beta,c_\gamma)\in (0,\infty)^2$ and $(\iota,\kappa)\in(0,1]^2$. Let $\alpha\in[0,1]$ be as in \ref{assume:Gen2}. Then, 
 \begin{enumerate}
     \item if  $\kappa(2-3\alpha/2)<\iota$, we have $\liminf_{n\rightarrow\infty}\|\nabla L_{\gamma_n}(\theta_n)\|=0$, $\P$-a.s,
     \item  if in addition $l(\theta)=\ell(\theta,u)$  for all $(\theta,u)\in \R^d\times \mathsf{U}$, then   $\liminf_{n\rightarrow\infty}\|\nabla L_{\gamma_n}(\theta_n)\|=0$ if  $ \kappa(1-\alpha/2)<\iota$ or if we have both $\kappa(1-\alpha/2)=\iota$ and $c_\beta c_\gamma^{\alpha/2 -1}<1/\bar{C}$, with the constant $\bar{C}<\infty$ as in Lemma \ref{lem:descent-general}.
 \end{enumerate}
\end{lemma}
 
\begin{proof}
For all $n\geq 1$ let $\delta_n$ be as defined in Lemma \ref{lem:descent-general}. To show the first part of the lemma remark that $\quad\sum_{n\geq 1}\beta_n=\infty$, as $\iota <1$, and, 
if $\kappa(2-3\alpha/2)<\iota$ then
\begin{align*}
\frac{\beta_n^2\gamma_n^{3\alpha/2-2}+\delta_n}{\beta_n} \to 0,\quad \text{as }n\to \infty.
\end{align*}

Then, under \ref{assume:Gen1} and by using Lemmas \ref{lemma:sequence}-\ref{lem:descent-general}, in particular  \eqref{eq:des_expect}, we check  that  $\liminf_{n\rightarrow\infty}\E\big[\|\nabla L_{\gamma_n}(\theta_n)\|\big]=0$, implying that  $\liminf_{n\rightarrow\infty}\|\nabla L_{\gamma_n}(\theta_n)\|=0$, $\P$-a.s.. Indeed,  $\liminf_{n\rightarrow\infty}\E\big[\|\nabla L_{\gamma_n}(\theta_n)\|\big]=0$ implies convergence to zero in $L^1$  of at least one subsequence of $(\E\sqrd{\|\nabla L_{\gamma_n}(\theta_n) \|})_{n\geq 1}$, in turn implying convergence in probability to zero of the same subsequence \citep[Exercise 4.13]{cinlar2011probability}. This further implies \citep[Theorem 3.3]{cinlar2011probability} that there exists a sub-subsequence converging to zero almost surely, in turn implying $\liminf_{n\rightarrow\infty}\|\nabla L_{\gamma_n}(\theta_n)\|=0$, $\P$-a.s.. 

To conclude the proof remark that under the assumptions of the second part of the lemma we have $\liminf_{n\geq 1}\big(1-\bar{C}\beta_n\gamma_n^{\frac{\alpha}{2}-1}\big)>0$. Then,  under   \ref{assume:Gen1} and by using   Lemmas \ref{lemma:sequence}-\ref{lem:descent-general}, it is readily checked that   $\liminf_{n\rightarrow\infty}\|\nabla L_{\gamma_n}(\theta_n)\|=0$. The proof of the lemma is complete.
\end{proof}

\subsubsection{From  convergence along a subsequence to convergence of the sequence}
We adapt the strategy developed in \cite{Patel2022} to our inhomogeneous case. 
\begin{lemma}\label{lem:Limsup-general}
Assume that \ref{assume:Gen2}  holds  and let $(\theta_n)_{n\geq 1}$ be as defined in \eqref{eq:theta_seq} where $\beta_n=c_\beta n^{-\iota}$ and $\gamma_n=c_\gamma n^{-\kappa}$ for all $n\geq 1$ and for some constants $(c_\beta,c_\gamma)\in (0,\infty)^2$ and $(\iota,\kappa)\in(0,1]^2$. Let $\alpha\in[0,1]$ and $\eta\geq 2$ be as in  \ref{assume:Gen2}, and assume that     $\min\{1-\kappa/2, \iota-\kappa(3/2-\alpha)\}>1/\eta$. Then, for all $\delta\in(0,\infty)$ there   exists $\P$-a.s.~an $n'\in\mathbb{N}$ such that either  $\|\nabla L_{\gamma_n}(\theta_n)\|> \delta $ for all $n\geq n'$ or  $\|\nabla L_{\gamma_n}(\theta_n)\|\leq \delta $  for all $n\geq n'$.
\end{lemma}
\begin{proof}

 Let $\bar{\gamma}_2>0$, $\eta\geq 2$ and $C:\mathsf{U}\rightarrow [1,\infty)$ be as in \ref{assume:Gen2}, and let $n_1 \in \Natural$ be such that $\gamma_n \leq \bar{\gamma}_2$ for all $n\geq n_1$. Remark that, under the assumptions on $(\gamma_n)_{n\geq 1}$,  there exists a constant $C_0\in[1,\infty)$ such that
\begin{align*}
\gamma_{n+1}^{-1+\frac{\alpha}{2}}\gamma_n^{\frac{\alpha-1}{2}}\leq C_0 \gamma_n^{\alpha-\frac{3}{2}},\quad  \gamma_n^{1/2}(\gamma_n/\gamma_{n+1})^{d/2}\frac{\gamma_n-\gamma_{n+1}}{\gamma_{n+1}^2}\leq C_0\,n^{\kappa/2-1},\quad\forall n\geq n_1.
\end{align*}  
Then, under \ref{assume:Gen2} (points 1, 2 and 4), for all $n\geq n_1$ and $u\in\mathsf{U}$ we have, $\P$-a.s.
\begin{equation*}
\begin{split}
\|\nabla \Ell_{\gamma_{n+1}}(\theta_{n+1},u)-\nabla \Ell_{\gamma_{n}}(\theta_{n},u)\|&\leq \|\nabla \Ell_{\gamma_{n+1}}(\theta_{n+1},u)-\nabla \Ell_{\gamma_{n+1}}(\theta_{n},u)\|\\
&\quad\quad+\|\nabla \Ell_{\gamma_{n+1}}(\theta_{n},u)-\nabla \Ell_{\gamma_{n}}(\theta_{n},u)\|\\
&\leq \gamma_{n+1}^{-1+\frac{\alpha}{2}}\|\theta_{n+1}-\theta_n\|C(u)+\gamma_n^{1/2}(\gamma_n/\gamma_{n+1})^{\frac{d }{2}}\frac{\gamma_n-\gamma_{n+1}}{\gamma_{n+1}^{2}} C(u)\\
&\leq \gamma_{n+1}^{-1+\frac{\alpha}{2}}\beta_n\|\nabla \Ell_{\gamma_{n}}(\theta_n,U_{n+1})\|C(u)+ C_0\,n^{\kappa/2-1} C(u)\\
&\leq \gamma_{n+1}^{-1+\frac{\alpha}{2}}\gamma_n^{\frac{\alpha-1}{2}}\beta_nC(U_{n+1})C(u)+ C_0\,n^{\kappa/2-1} C(u)\\
&\leq C_0C(U_{n+1})C(u)\big(\gamma_n^{\alpha-\frac{3}{2}}\beta_n+  n^{\kappa/2-1}\big).
\end{split}
\end{equation*}

Therefore,  noting that $(x+y)^\eta\leq 2^{\eta-1}(x^\eta+y^\eta)$ for all $x,y\in\R$, as $\eta\geq 2$, it follows that for all $n\geq n_1$ we have,   $\P$-a.s.
\begin{align}\label{eq:BB_grad}
\|\nabla L_{\gamma_{n+1}}(\theta_{n+1})-\nabla L_{\gamma_{n}}(\theta_{n})\|^{\eta}&\leq  \bar{C}C(U_{n+1})^{\eta}\delta_{n},\quad\delta_{n}=n^{-\eta(\iota+\kappa(\alpha-3/2))} +n^{-\eta(1-\kappa/2)}
\end{align}
where $\bar{C}=2^{\eta-1}C_0^\eta\E[C(U)^{\eta }]<\infty$. We now let $(\delta,\epsilon)\in(0,\infty)$. Then, for all $n\geq n_1$ we have
\begin{equation}\label{eq:lem_sup2}
\begin{split}
\P\Big(\|\nabla &L_{\gamma_{n+1}}(\theta_{n+1})\|\geq \delta+\epsilon,\,\,\|\nabla L_{\gamma_{n}}(\theta_{n})\|\leq \delta\Big)\\
&=\P\Big(\|\nabla L_{\gamma_{n+1}}(\theta_{n+1})\|-\|\nabla L_{\gamma_{n}}(\theta_{n})\|+\|\nabla L_{\gamma_{n}}(\theta_{n})\|\geq\delta+\epsilon,\,\,\|\nabla L_{\gamma_{n}}(\theta_{n})\|\leq \delta\Big)\\
&\leq \P\Big(\|\nabla L_{\gamma_{n+1}}(\theta_{n+1})\|-\|\nabla L_{\gamma_{n}}(\theta_{n})\|\geq\epsilon,\,\,\|\nabla L_{\gamma_{n}}(\theta_{n})\|\leq \delta\Big)\\
&\leq \P\Big(\|\nabla L_{\gamma_{n+1}}(\theta_{n+1})\|-\|\nabla L_{\gamma_{n}}(\theta_{n})\|\geq\epsilon\Big)\\
&\leq \P\Big(\|\nabla L_{\gamma_{n+1}}(\theta_{n+1})-\nabla L_{\gamma_{n}}(\theta_{n})\|\geq\epsilon\Big)\\
&\leq \frac{\E\Big[\|\nabla L_{\gamma_{n+1}}(\theta_{n+1})-\nabla L_{\gamma_{n}}(\theta_{n})\|^{\eta}\Big]}{\epsilon^\eta}\\
&\leq  \bar{C}^2\delta_{n} \epsilon^{-\eta} 
\end{split}
\end{equation}
where the third inequality holds by the reverse triangle inequality,  the fourth inequality holds by Markov's inequality and the last inequality holds by \eqref{eq:BB_grad} (noting that $\bar{C}>\E[C(U)^{\eta}]$).

Under the assumptions on $\kappa$ and $\iota$ we have  $\sum_{n\geq 1}\delta_n<\infty$ and thus, using \eqref{eq:lem_sup2},
\begin{align*}
\sum_{n\geq 1}\P\Big(\|\nabla L_{\gamma_{n+1}}(\theta_{n+1})\|\geq&\delta+\epsilon,\,\,\|\nabla L_{\gamma_{n}}(\theta_{n})\|\leq \delta\Big)<\infty.
\end{align*}
From the Borel-Cantelli lemma it follows that the set
\begin{align*}
\Omega_{\delta,\epsilon}:=\Big\{ \|\nabla L_{\gamma_{n+1}}(\theta_{n+1})\|\geq \delta+\epsilon\text{ and }\|\nabla L_{\gamma_{n}}(\theta_{n})\|\leq \delta\,\,\,i.o.\Big\}
\end{align*}
is such that $\P(\Omega_{\delta,\epsilon})=0$ and the result of the lemma follows upon noting that
\begin{align*}
\Big\{\|\nabla L_{\gamma_{n+1}}(\theta_{n+1})\|> \delta\text{ and }\|\nabla L_{\gamma_{n}}(\theta_{n})\|\leq \delta\,\,\,\,i.o.\Big\}=\bigcup_{k\in\mathbb{N}}\Omega_{\delta,1/k}.
\end{align*}
\end{proof}

\subsection{Proof of Propositions \ref{prop:2point}-\ref{prop:Bayes}}\label{proof_prop12}

Noting that when $\Ell_\gamma=\bar{\ell}_\gamma$, \ref{assume:Gen1} trivially holds under \ref{assume:lower-bound}, Proposition \ref{prop:2point} is established once Lemma \ref{lemma:standard-gaussian_smoothing} is proved. Proposition \ref{prop:Bayes} is a direct consequence of  Lemmas \ref{lemma:S1forlgamma}-\ref{lemma:estimates-laplace}. While Lemma \ref{lemma:standard-gaussian_smoothing} is relatively easy to establish, Lemma \ref{lemma:estimates-laplace} requires much more effort due to the need to control a ratio of integrals.

\subsubsection{Assumption \ref{assume:Gen2} holds for \texorpdfstring{$\Ell_\gamma=\bar{\ell}_\gamma$}{Lg} under \ref{assume:holder} (and thus under \ref{assume:holder2})}

 \begin{lemma} 
\label{lemma:standard-gaussian_smoothing}
 Assume that \ref{assume:holder} holds. Then,  there exists a  constant $C\in(0,\infty)$ such that, for all $\theta,\theta'\in \R^d$, all $u\in \mathsf{U}$  and all $0< \tilde{\gamma}\leq \gamma\leq 1$, and with $\alpha$ and $J(\cdot)$ as in  \ref{assume:holder}, we have:
 \begin{enumerate}
     \item $\|\nabla \bar{\ell}_{\gamma}(\theta,u)\|\leq C\,J(u)\gamma^{\frac{\alpha-1}{2}}$,
     \item  $\left\|\nabla\bar{\ell}_{\gamma}(\theta,u)- \nabla\bar{\ell}_{\gamma}(\theta',u)\right\|\leq C\,J(u) \gamma^{-1+\frac{\alpha}{2}}\|\theta-\theta'\|$,
     \item $\left|\bar{\ell}_{\tilde{\gamma}}(\theta,u)-\bar{\ell}_{\gamma}(\theta,u)\right|\leq C\,J(u) ( \gamma/\tilde{\gamma})^{\frac{d}{2}}\frac{\gamma-\tilde{\gamma}}{\tilde{\gamma}}$,
     \item $\|\nabla \bar{\ell}_{\tilde{\gamma}}(\theta,u)-\nabla\bar{\ell}_{\gamma}(\theta,u) \|\leq C\,J(u) \gamma^{1/2}( \gamma/\tilde{\gamma})^{\frac{d}{2}}\frac{\gamma-\tilde{\gamma}}{\tilde{\gamma}^{2}}$.
 \end{enumerate}
 \end{lemma}
 
 \begin{proof}
  Below we let   $0< \tilde{\gamma}\leq \gamma\leq \bar{\gamma}_2 = 1$, $\theta\in\R^d$,  $u\in \mathsf{U}$ and $i\in\{1,\dots,d\}$ be arbitrary.
  
  For the first point of the lemma,   using \eqref{eq:ll} we can write
   \begin{align*}
       \left|\frac{\partial }{\partial \theta_i} \bar{\ell}_{\gamma}(\theta,u)\right|&=\gamma^{-1/2}\left|\int_{\R^d} \ell(\theta+\gamma^{1/2}z,u)z_i\phi_{d,1}(z)\dd z\right|\\
       &=\gamma^{-1/2}\left|\int_{\R^d} \round{\ell(\theta+\gamma^{1/2}z,u)-\ell(\theta,u)}z_i\phi_{d,1}(z)\dd z\right|\\
       &\leq \gamma^{-1/2}\int_{\R^d} \left|\ell(\theta+\gamma^{1/2}z,u)-\ell(\theta,u)\right|\left|z_i\right|\phi_{d,1}(z)\dd z\\
       &\leq 2\gamma^{\frac{\alpha-1}{2}}J(u)\int_{\R^d}(1+ \|z\|^{\beta})\left|z_i\right|\phi_{d,1}(z)\dd z
   \end{align*}
   and thus  the first part of the lemma holds with $C\geq2d^{1/2}\int_{\R^d} (1+\|z\|^{\beta})\left|z_1\right|\phi_{d,1}(z)\dd z<\infty$.
   
   For the second part of the lemma,   using \eqref{eq:ll} as  well as   $\int_{\R} (y^2-1)\phi_{1,1}(y)\mathrm{d}y=0$, we can write
\begin{equation}\label{eq:int1}
\begin{split}
    \left|\frac{\partial^2}{\partial \theta^2_i}\bar{\ell}_{\gamma}(\theta,u)\right|&=\left|\gamma^{-1}\int_{\R^d} \ell(\theta+\gamma^{1/2}z,u)\round{z_i^2-1}\phi_{d,1}(z)\dd z\right|\\
    &=\left|\gamma^{-1}\int_{\R^d} \round{\ell(\theta+\gamma^{1/2}z,u)-\ell(\theta,u)}\round{z_i^2-1}\phi_{d,1}(z)\dd z\right|\\
    &\leq \gamma^{-1}\int_{\R^d} \left|\ell(\theta+\gamma^{1/2}z,u)-\ell(\theta,u)\right|\left|z_i^2-1\right|\phi_{d,1}(z)\dd z\\
    &\leq \gamma^{-1+\frac{\alpha}{2}}J(u)C'
\end{split}
\end{equation}
with $C'=2\int_{\R^d} \left\{1+\|z\|^{\beta}\right\}\left|z_i^2-1\right|\phi_{d,1}( z)\dd z<\infty$. Moreover, for  any $j\neq i$ we have, using similar calculations and  noting that  $\int_{\R^d} z_iz_j\phi_{d,1}(\mathrm{d}z)=0$,
\begin{equation}\label{eq:int2}
\begin{split}
    \left|\frac{\partial^2}{\partial \theta_i\partial \theta_j}\bar{\ell}_{\gamma}(\theta,u)\right|&=\left|\gamma^{-1}\int_{\R^d} \ell(\theta+\gamma^{1/2}z,u)z_iz_j\phi_{d,1}(z)\dd z\right|\\
    &=\left|\gamma^{-1}\int_{\R^d} \round{\ell(\theta+\gamma^{1/2}z,u)-\ell(\theta,u)}z_iz_j\phi_{d,1}(z)\dd z\right|\\
    &\leq \gamma^{-1}\int_{\R^d} \left|\ell(\theta+\gamma^{1/2}z,u)-\ell(\theta,u)\right|\left|z_i\right|\left|z_j\right|\phi_{d,1}(z)\dd z\\
    &\leq \gamma^{-1+\frac{\alpha}{2}}J(u)C''
\end{split}
\end{equation}
with $C''=2\int_{\R^d}(1+ \|z\|^{\beta})\left|z_i\right|\left|z_j\right|\phi_{d,1}( z)\dd z$. By combining \eqref{eq:int1} and  \eqref{eq:int2}, it follows that the second part of the lemma holds with $C\geq (C'+C'')d$, from which we deduce that the second part of the lemma holds.

Next, to prove the third point of the lemma,  we remark that  by Corollary~\ref{cor:delta-phi-gam} (see technical results in Appendix~\ref{subsec:from-gen-to-example}) and  using \eqref{eq:ll}, for some constant $C'<\infty$ we have
\begin{align*}
|\bar{\ell}_{\gamma}(\theta,u)-\bar{\ell}_{\tilde{\gamma}}(\theta,u)|
&=\bigg|\int_{\R^d} \big(\ell(\theta+z,u)-\ell(\theta,u)\big)\big[\phi_{d,\gamma}(z)-\phi_{d,\tilde{\gamma}}(z)\big]\mathrm{d}z\bigg|\\
&\leq C'\left(\frac{\gamma}{\tilde{\gamma}}\right)^{d/2}\frac{\gamma-\tilde{\gamma}}{\tilde{\gamma}} \int_{\R^d} \big|\ell(\theta+\gamma^{1/2} z,u)-\ell(\theta,u)\big|(1+\|z\|^{2}) \phi_{d,1}(z) \mathrm{d}z \\
&\leq 2C'\left(\frac{\gamma}{\tilde{\gamma}}\right)^{d/2}\frac{\gamma-\tilde{\gamma}}{\tilde{\gamma}}J(u)\int_{\R^d}(1+\|z\|^{\beta})(1+\|z\|^{2})\phi_{d,1}(z)\dd z 
\end{align*}
showing  that the third part of the lemma holds with $C\geq2C'\int_{\R^d}(1+\|z\|^{\beta})(1+\|z\|^{2})\phi_{d,1}( z)\dd z<\infty$.

Finally, to show the last part of the lemma holds,  we remark that by Corollary~\ref{cor:delta-phi-gam} and using \eqref{eq:ll}, for some constant $C'<\infty$ we have
\begin{align*}
\Big|\frac{\partial }{\partial \theta_i}\bar{\ell}_{\gamma}(\theta,u)-\frac{\partial }{\partial \theta_i}\bar{\ell}_{\tilde{\gamma}}(\theta,u)\Big| & =\bigg|\int_{\R^d}\ell(\theta+z,u)z_{i}\left[\gamma^{-1}\phi_{d,\gamma}(z)-\tilde{\gamma}^{-1}\phi_{d,\tilde{\gamma}}(z)\right]\mathrm{d}z \bigg|\\
 & =\bigg|\int_{\R^d}\left[\ell(\theta+z,u)-\ell(\theta,u)\right]z_{i}\left[\gamma^{-1}\phi_{d,\gamma}(z)-\tilde{\gamma}^{-1}\phi_{d,\tilde{\gamma}}(z)\right]\mathrm{d}z\bigg|\\
 &\leq  \gamma^{1/2}\,C'\left(\frac{\gamma}{\tilde{\gamma}}\right)^{d/2}\frac{\gamma-\tilde{\gamma}}{\tilde{\gamma}^2}\int_{\R^d}\big|\ell(\theta+\gamma^{1/2}z,u)-\ell(\theta,u)\big||z_i|(1+\|z\|^2)\phi_{d,1}(z)\dd z\\
 &\leq  \gamma^{1/2} 2C'\left(\frac{\gamma}{\tilde{\gamma}}\right)^{d/2}\frac{\gamma-\tilde{\gamma}}{\tilde{\gamma}^2}J(u)\int_{\R^d} (1+\|z\|^{\beta})|z_i|(1+\|z\|^2)\phi_{d,1}(z) \mathrm{d}z
\end{align*}
showing that last part of the lemma holds with $C\geq 2C'd^{1/2}\int_{\R^d} (1+\|z\|^{\beta})|z_1|(1+\|z\|^2)\phi_{d,1}(z) \mathrm{d}z<\infty$.
 \end{proof}

 \subsubsection{Assumption \ref{assume:Gen1} holds for \texorpdfstring{$\Ell_\gamma=\ell_\gamma$}{Lg} under \ref{assume:lower-bound} and \ref{assume:holder2}}
  
\begin{lemma} \label{lemma:S1forlgamma}
Assume that \ref{assume:lower-bound} and \ref{assume:holder2} hold. Then, \ref{assume:Gen1} holds for $\Ell_\gamma=\ell_\gamma$.
\end{lemma}
\begin{proof}
Let $\theta\in\R^d$, $u\in\mathsf{U}$ and $\gamma\in(0,1]$ be arbitrary. Then, using the Bayes-Laplace Sandwich
Theorem, see Theorem \ref{thm: bayes-laplace sandwich} in Appendix~\ref{app:sandwich}, we have
 \begin{equation}\label{eq:BL1}
 \begin{split}
 \ell_{\gamma}(\theta,u)&\geq  \int_{\R^d} \ell(\theta+z,u)\frac{e^{-\ell(\theta+z,u)}\phi_{d,\gamma}(z)}{\int_{\R^d} e^{-\ell(\theta+z,u)}\phi_{d,\gamma}(z )\mathrm{d}z }\mathrm{d}z\\
 &= \ell(\theta,u)+\int_{\R^d} \Big(\ell(\theta+z,u)-\ell(\theta,u)\Big)\frac{e^{-\ell(\theta+z,u)}\phi_{d,\gamma}(z)}{\int_{\R^d} e^{-\ell(\theta+z,u)}\phi_{d,\gamma}(z )\mathrm{d}z }\mathrm{d}z
\end{split}
\end{equation}
where, under \ref{assume:holder2} and using  \eqref{eq:ll},
 \begin{equation}\label{eq:BL2}
 \begin{split}
\bigg|\int_{\R^d} \Big(\ell(\theta+z,u)-\ell(\theta,u)\Big)&\frac{e^{-\ell(\theta+z,u)}\phi_{d,\gamma}(z)}{\int_{\R^d} e^{-\ell(\theta+z,u)}\phi_{d,\gamma}(z )\mathrm{d}z }\mathrm{d}z\bigg|\\
&\leq \int_{\R^d} \big|\ell(\theta+\gamma^{1/2} z,u)-\ell(\theta,u)\big)\big|\frac{e^{-\ell(\theta+\gamma^{1/2} z,u)}\phi_{d,1}(z)}{\int_{\R^d} e^{-\ell(\theta+\gamma^{1/2} z,u)}\phi_{d,1}(z )\mathrm{d}z }\mathrm{d}z\\
&\leq 2 J(u)\int_{\R^d}(2+\|z\|^2)\frac{e^{-\ell(\theta+z,u)}\phi_{d,1}(z)}{\int_{\R^d} e^{-\ell(\theta+z',u)}\phi_{d,1}(z )\mathrm{d}z }\mathrm{d}z\\
&\leq 2 J(u)\int_{\R^d}(2+\|z\|^2) e^{4\gamma^{1/\alpha}J(u)(1+\|z\|^\beta)}\phi_{d,1}(z )\mathrm{d}z\\
&\leq 6 J(u)\int_{\R^d}(1+\|z\|^4) e^{4\gamma^{1/\alpha}J(u)(1+\|z\|^\beta)}\phi_{d,1}(z )\mathrm{d}z\\
&=:\tilde{G}_\gamma(u)
\end{split}
\end{equation} 
where the  last inequality holds by Lemma \ref{lemma:zero_mean} part 2. By Lemma \ref{lemma:assume}, there exists a constant $\bar{\gamma}'\in(0,1]$ such that $\E[\tilde{G}_{\bar{\gamma}'}(U)]<\infty$  and thus, by using \eqref{eq:BL1}-\eqref{eq:BL2} we have, under \ref{assume:lower-bound},
\begin{align*}
\inf_{\theta\in\R^d}\E\big[\ell_{\gamma}(\theta,U)\big]\geq c:=\inf_{\theta\in\R^d} \E[\ell(\theta,U)]-\E[\tilde{G}_{\bar{\gamma}}(U)]>-\infty,\quad\forall\gamma\in (0,\bar{\gamma}].
\end{align*}
The proof of the lemma is complete.
 \end{proof}

\subsubsection{Assumption \ref{assume:Gen2} holds for \texorpdfstring{$\Ell_\gamma=\ell_\gamma$}{Lg} under   \ref{assume:holder2}}

\begin{lemma}\label{lemma:estimates-laplace}
Assume that \ref{assume:holder2} holds. Then, for all $p\in[1,\infty)$, there exists a $\bar{\gamma}_p\in (0,\infty)$ and a  function $G_p:\mathsf{U}\rightarrow[1,\infty)$ such that $\E[G_p(U)^p]<\infty$ and such that, for all $\theta,\theta'\in\R^d$, all $u\in\mathsf{U}$ and all $0<\tilde{\gamma}\leq \gamma\leq \bar{\gamma}_p$  we have, with $\alpha$ as in \ref{assume:holder2},
\begin{enumerate}
\item $\|\nabla \ell_\gamma(\theta,u)\|\leq G_p(u)\gamma^{\frac{\alpha-1}{2}}$,
\item $\|\nabla \ell_\gamma(\theta,u)-\nabla \ell_\gamma(\theta',u)\|\leq G_p(u)\gamma^{-1+\frac{\alpha}{2}}\|\theta-\theta'\| $,
\item $\big|\ell_{\tilde{\gamma}}(\theta,u)-\ell_\gamma(\theta,u)\big|\leq G_p(u) ( \gamma/\tilde{\gamma})^{\frac{d}{2}} \frac{\gamma-\tilde{\gamma}}{\tilde{\gamma}}$,
\item $\|\nabla \ell_\gamma(\theta,u)-\nabla \ell_{\tilde{\gamma}}(\theta,u)\|\leq  G_p(u)\gamma^{1/2}( \gamma/\tilde{\gamma})^{\frac{d}{2}} \frac{\gamma-\tilde{\gamma}}{\tilde{\gamma}^{2}}$.
\end{enumerate}
\end{lemma}

\begin{proof}
Below we let $\theta\in\R^d$, $u\in\mathsf{U}$, $i\in\{1,\dots,d\}$, $p\in [1,\infty)$ and $0<\tilde{\gamma}\leq \gamma\leq \bar{\gamma}_p$ be arbitrary,  where $\bar{\gamma}_p=(\bar{c}_{2p}/4)^{2/\alpha}$ with  $\bar{c}_{2p}\in(0,\infty)$  as in Lemma~\ref{lemma:assume}.  
Define
\begin{align*}
\tilde{G}_p(\tilde{u}):= 4 J(\tilde{u}) \int_{\R^d}(4+ \|z\|^4)  e^{\bar{c}_{2p}J(\tilde{u})(1+\|z\|^\beta)} \phi_{d,1}(z)\dd z \geq 1,\quad\forall \tilde{u}\in\mathsf{U},
\end{align*} 
which is such that $\E[\tilde{G}_p(U)^{2p}]<\infty$ from Lemma~\ref{lemma:assume},
and  let
\begin{align*}
g_\gamma(\theta,u)=\int_{\R^d} e^{-\ell(x,u)}\phi_{d,\gamma }(x-\theta)\dd x=\int_{\R^d} e^{-\ell(\theta+\gamma^{1/2}z,u)}\phi_{d,1}(z)\dd z.
\end{align*}
Remark that  
\begin{equation}\label{eq:deriv1}
\begin{split}
\partial_i g_{\gamma}(\theta,u):=\frac{\partial}{\partial\theta_i}g_\gamma(\theta,u)=\gamma^{-1/2}\int_{\R^d} e^{-\ell(\theta+\gamma^{1/2}z,u)}z_i\phi_{d,1}(z)\dd z
\end{split}
\end{equation}
and that, for all $j\in \{1,\dots,d\}$, 
\begin{equation}\label{eq:deriv2}
\begin{split}
\partial_{ji}g_{\gamma}(\theta,u)&:=\frac{\partial}{\partial\theta_j}\partial_i g_{\gamma}(\theta,u)=-\mathbf{1}_{\{j\}}(i)\frac{1}{\gamma }g_{\gamma}(\theta,u)+\gamma^{-1}\int_{\R^d} e^{-\ell(\theta+\gamma^{1/2}z,u)}z_j z_i \phi_{d,1}(z)\dd z.
\end{split}
\end{equation}

To show the first part of the lemma remark that, using \eqref{eq:deriv1},
\begin{equation}\label{eq:deriv_w1}
\begin{split}
\Big|\frac{\partial  \ell_\gamma(\theta,u)}{\partial\theta_i}\Big| =\Big|\frac{\partial_i g_{\gamma}(\theta,u)}{ g_{\gamma}(\theta,u)}\Big|&=\gamma^{-1/2}\frac{\big|\int_{\R^d} e^{ -\ell(\theta+\gamma^{1/2}z,u)}z_i\phi_{d,1}(z)\dd z\big|}{\int_{\R^d} e^{ -\ell(\theta+\gamma^{1/2}z,u)} \phi_{d,1}(z)\dd z} \\
&\leq 2J(u)\gamma^{\frac{\alpha-1}{2}} \int_{\R^d}(2+\|z\|^2) \|z\| e^{4\gamma^{\alpha/2}J(u)(1+\|z\|^\beta)}\phi_{d,1}(z)\dd z\\
&\leq 4J(u)\gamma^{\frac{\alpha-1}{2}} \int_{\R^d}(4+\|z\|^4) e^{4\gamma^{\alpha/2}J(u)(1+\|z\|^\beta)} \phi_{d,1}(z)\dd z\\
&\leq  \gamma^{\frac{\alpha-1}{2}}\tilde{G}_{p}(u)
\end{split}
\end{equation}
where the first inequality holds by Lemma~\ref{lemma:zero_mean} (third part) and uses the fact that $|z_i| \leq \|z\|$, the second inequality uses that
$(2+\|z\|^2)\|z\|\leq (2+\|z\|^2)^2  \leq 2(4+\|z\|^4)$
 and the last inequality uses the fact that  $4\gamma^{\alpha/2}\leq \bar{c}_{2p}$. The result in the first part of the lemma follows with $G_p(u)\geq d^{1/2} \tilde{G}_{p}(u)$.

To show the second part of the lemma    let  $j\in\{1,\dots,d\}$  and note that
\begin{equation}\label{eq:deriv2_w1}
\begin{split}
\Big|\frac{\partial^2  \ell_\gamma(\theta,u)}{\partial\theta_i\partial\theta_j}\Big|
&=\bigg|\frac{g_{\gamma}(\theta,u)\partial_{ji}g_{\gamma}(\theta,u)-\big(\partial_i g_{\gamma}(\theta,u)\big)\big(\partial_jg_{\gamma}(\theta,u)\big)}{g_{\gamma}(\theta,u)^2}\bigg|\\
&\leq \bigg|\frac{\partial_{ji}g_{\gamma}(\theta,u)}{ g_{\gamma}(\theta,u)}\bigg|+\Big|\frac{\partial  \ell_\gamma(\theta,u)}{\partial\theta_i}\Big|\,\Big|\frac{\partial  \ell_\gamma(\theta,u)}{\partial\theta_j}\Big|.
\end{split}
\end{equation}
Assume first that $j=i$. In this case we have, using \eqref{eq:deriv2},
\begin{equation}\label{eq:deriv2_w2}
\begin{split}
\bigg|\frac{\partial_{ji}g_{\gamma}(\theta,u)}{ g_{\gamma}(\theta,u)}\bigg|&=\frac{1}{\gamma}\frac{\big|\int_{\R^d} e^{-\ell(\theta+\gamma^{1/2}z,u)}(z_i^2-1)\phi_{d,1}(z)\dd z\big|}{\int_{\R^d} e^{-\ell(\theta+\gamma^{1/2}z,u)} \phi_{d,1}(z)\dd z}\\
&\leq 2J(u)\gamma^{\frac{\alpha-2}{2}} \int_{\R^d}(2+\|z\|^2) (1+\|z\|^2)e^{4\gamma^{\alpha/2}J(u)(1+\|z\|^\beta)} \phi_{d,1}(z)\dd z\\
&\leq   4J(u)\gamma^{\frac{\alpha-2}{2}} \int_{\R^d}(4+\|z\|^4) e^{4\gamma^{\alpha/2}J(u)(2+\|z\|^\beta)} \phi_{d,1}(z)\dd z\\
&\leq     \gamma^{\frac{\alpha-2}{2}} \tilde{G}_p(u)\,,
\end{split}
\end{equation}
where the first inequality holds    by Lemma~\ref{lemma:zero_mean} (third part) and uses the fact that $|z_i^2-1|\leq 1+\|z\|^2$,   the second inequality uses that 
$   (2+\|z\|^2) (1+\|z\|^2) \leq (2+\|z\|^2)^2 \leq 2 (4+\|z\|^4)$ and the last inequality holds since     $4\gamma^{\alpha/2}\leq \bar{c}_{2p}$. On the other hand, if $j\neq i$ we have, using \eqref{eq:deriv2},
\begin{equation}\label{eq:deriv2_w22}
\begin{split}
\bigg|\frac{\partial_{ji}g_{\gamma}(\theta,u)}{ g_{\gamma}(\theta,u)}\bigg|&=\frac{1}{\gamma}\frac{\big|\int_{\R^d} e^{-\ell(\theta+\gamma^{1/2}z,u)}z_iz_j\phi_{d,1}(z)\dd z\big|}{\int_{\R^d} e^{-\ell(\theta+\gamma^{1/2}z,u)} \phi_{d,1}(z)\dd z}\\
&\leq J(u)\gamma^{\frac{\alpha-2}{2}} \int_{\R^d}(2+\|z\|^2) \|z\|^2e^{4\gamma^{\alpha/2}J(u)(1+\|z\|^\beta)} \phi_{d,1}(z)\dd z\\
&\leq 2 J(u)\gamma^{\frac{\alpha-2}{2}} \int_{\R^d}(4+\|z\|^4) e^{4\gamma^{\alpha/2}J(u)(1+\|z\|^\beta)} \phi_{d,1}(z)\dd z\\
&\leq  \gamma^{\frac{\alpha-2}{2}} \tilde{G}_p(u)\, ,
\end{split}
\end{equation}
where the first inequality holds    by Lemma~\ref{lemma:zero_mean} (third part) and uses that $|z_iz_j| \leq (|z_i|^2+|z_j|^2)/2\leq \|z\|^2/2$,   the second inequality follows from $(2+\|z\|^2)\|z\|^2 \leq  (2+\|z\|^2)^2 \leq 2(4+\|z\|^4)$  and   the last inequality holds since     $4\gamma^{\alpha/2}\leq \bar{c}_{2p}$. By combining \eqref{eq:deriv_w1}-\eqref{eq:deriv2_w22}, we obtain that
\begin{align*}
\Big|\frac{\partial  \ell^2_\gamma(\theta,u)}{\partial\theta_i\partial\theta_j}\Big|\leq   \gamma^{\frac{\alpha-2}{2}} \tilde{G}_p(u)+\gamma^{\alpha-1}\tilde{G}_{p}(u)^2
\end{align*}
and the result in the second part of the lemma follows from Taylor's theorem and taking $G_p(u)=C_0 \tilde{G}_{p}(u)^2$ for a constant $C_0>0$ sufficiently large.

To show the third part of the lemma   assume first that $\ell_{\tilde{\gamma}}(\theta,u)-\ell_\gamma(\theta,u)\geq 0$. Then, using the fact that $\log(1+x)\leq x$ for all $x\in (-1,\infty)$,
\begin{equation}\label{eq:bound_l}
\begin{split}
|\ell_{\tilde{\gamma}}(\theta,u)-\ell_{\gamma}(\theta,u)|&=\log\bigg(1+ \frac{\int_{\R^d}e^{-\ell(\theta+z,u)}\big(\phi_{d,\gamma}(  z)-\phi_{d,\tilde{\gamma}}(  z)\big)\dd z}{\int_{\R^d}e^{-\ell(\theta+z,u)}\phi_{d,\tilde{\gamma}}(  z)\dd z}\bigg)\\
&\leq \frac{\int_{\R^d}e^{ -\ell(\theta+z,u)}\big(\phi_{d,\gamma}( z)-\phi_{d,\tilde{\gamma}}(  z)\big)\dd z}{\int_{\R^d}e^{ -\ell(\theta+z,u)}\phi_{d,\tilde{\gamma}}(\dd z)}\\
& \leq  C(\gamma/\tilde{\gamma})^{\frac{d}{2}} \frac{\gamma-\tilde{\gamma}}{\tilde{\gamma}}  \frac{\int_{\R^d}e^{ -\ell(\theta+\gamma^{1/2} z,u)}(1+ \|z\|^2)\phi_{d,1}(z)\dd z}{\int_{\R^d}e^{ -\ell(\theta+\tilde{\gamma}^{1/2}z,u)}\phi_{d,1}( z)\dd z}
\end{split}
\end{equation}
where, by Corollary \ref{cor:delta-phi-gam}, the second inequality holds for some constant $C<\infty$. Together with Lemma~\ref{lemma:zero_mean} (second part), and noting that  $1+\|z\|^2 \leq (2+\|z\|^4)$ and   that $4\gamma^{\alpha/2}\leq \bar{c}_{2p}$, it follows that
\begin{equation}\label{eq:bound_l_p1}
\begin{split}
|\ell_{\tilde{\gamma}}(\theta,u)-\ell_{\gamma}(\theta,u)
| &\leq   C   (\gamma/\tilde{\gamma})^{\frac{d}{2}} \frac{\gamma-\tilde{\gamma}}{\tilde{\gamma}} \int_{\R^d}e^{4J(u)\gamma^{\alpha/2}(1+\|z\|^\beta)}(2+\|z\|^4)\phi_{d,1}( z)\dd z\\
 &\leq   C   (\gamma/\tilde{\gamma})^{\frac{d}{2}} \frac{\gamma-\tilde{\gamma}}{\tilde{\gamma}}  \tilde{G}_p(u).
\end{split}
\end{equation}
 Assume now that $\ell_{\tilde{\gamma}}(\theta,u)-\ell_\gamma(\theta,u)\leq 0$. Then, by following similar calculations as in \eqref{eq:bound_l}-\eqref{eq:bound_l_p1}, for some constant $C<\infty$ we have
\begin{equation}\label{eq:bound_l_p2}
\begin{split}
|\ell_{\gamma}(\theta,u)-\ell_{\tilde{\gamma}}(\theta,u)|
&\leq  \frac{\int_{\R^d}e^{ -\ell(\theta+z,u)}\big(\phi_{d,\tilde{\gamma}}( z)-\phi_{d, \gamma }(  z)\big)\dd z}{\int_{\R^d}e^{ -\ell(\theta+z,u)}\phi_{d, \gamma }(\dd z)}\leq  C  (\gamma/\tilde{\gamma})^{\frac{d}{2}} \frac{\gamma-\tilde{\gamma}}{\tilde{\gamma}}  \tilde{G}_p(u).
\end{split}
\end{equation}
By combining \eqref{eq:bound_l_p1}-\eqref{eq:bound_l_p2} we obtain
\begin{align}\label{eq:bound-delta-l_gam-proof}
|\ell_{\gamma}(\theta,u)-\ell_{\tilde{\gamma}}(\theta,u)|\leq \frac{|g_\gamma(\theta,u)-g_{\tilde{\gamma}}(\theta,u)|}{\min\{g_\gamma(\theta,u),g_{\tilde{\gamma}}(\theta,u)\}}\leq  C  (\gamma/\tilde{\gamma})^{\frac{d}{2}} \frac{\gamma-\tilde{\gamma}}{\tilde{\gamma}}  \tilde{G}_p(u)
\end{align}
and   the third part of the lemma follows.

To show the last part of the lemma  note first that
\begin{equation}\label{eq:last_part}
\begin{split}
\Big|\frac{\partial  \ell_{\gamma}(\theta,u)}{\partial\theta_i}-\frac{\partial  \ell_{\tilde{\gamma}}(\theta,u)}{\partial\theta_i}\Big|&\leq \Big|\frac{\partial  \ell_{\tilde{\gamma}}(\theta,u)}{\partial\theta_i}\Big|\frac{\big|g_{\gamma}(\theta,u)- g_{\tilde{\gamma}}(\theta,u)\big|}{g_{\gamma}(\theta,u)}\Big|+\frac{\big|\partial_i g_{\gamma}(\theta,u)-\partial_i g_{\tilde{\gamma}}(\theta,u)\big|}{g_{\gamma}(\theta,u)}.
\end{split}
\end{equation}
Noting that
\begin{equation}\label{eq:last_part_21}
\begin{split}
\Big|\gamma^{-1/2}\int_{\R^d} e^{\ell(\theta,u)-\ell(\theta+\gamma^{1/2}z,u)}z_i\phi_{d,1}( z)\dd z&-\tilde{\gamma}^{-1/2}\int_{\R^d} e^{\ell(\theta,u)-\ell(\theta+\tilde{\gamma}^{1/2}z,u)}z_i\phi_{d,1}(  z)\dd z\Big|\\
&=\Big|\int_{\R^d} e^{\ell(\theta,u)-\ell(\theta+z,u)}z_i\Big(\gamma^{-1}\phi_{d,\gamma}( z)-\tilde{\gamma}^{-1}\phi_{d,\tilde{\gamma}}(  z)\Big)\dd z\Big| 
\end{split}
\end{equation}
and using \eqref{eq:deriv1}, it follows Corollary~\ref{cor:delta-phi-gam} that for some constant $C<\infty$ we have
\begin{align*}
    \big|\partial_i g_{\gamma}(\theta,u)-\partial_i g_{\tilde{\gamma}}(\theta,u)\big| &\leq C(\gamma/\tilde{\gamma})^{\frac{d}{2}} \frac{\gamma-\tilde{\gamma}}{\tilde{\gamma}^2}\gamma^{1/2}\int_{\R^d}e^{\ell(\theta,u)-\ell(\theta+\gamma^{1/2} z,u)}|z_i| (1+\|z\|^2) \phi_{d,1}(\dd z)\\
    & \leq C(\gamma/\tilde{\gamma})^{\frac{d}{2}} \frac{\gamma-\tilde{\gamma}}{\tilde{\gamma}^2}\gamma^{1/2}\int_{\R^d}e^{\ell(\theta,u)-\ell(\theta+\gamma^{1/2} z,u)} (\|z\|+\|z\|^3)\phi_{d,1}(\dd z).
\end{align*}
Using this latter result, noting that $(\|z\|+\|z\|^3) \leq (2+\|z\|^4)$ and that $4\gamma^{\alpha/2}\leq \bar{c}_{2p}$,  and applying Lemma~\ref{lemma:zero_mean} (second part), we obtain
\begin{equation} \label{eq:bound-nabla-g_gam-proof-2}
\frac{\big|\partial_i g_{\gamma}(\theta,u)-\partial_i g_{\tilde{\gamma}}(\theta,u)\big|}{g_{\gamma}(\theta,u)} \leq  C \gamma^{1/2}(\gamma/\tilde{\gamma})^{\frac{d}{2}} \frac{\gamma-\tilde{\gamma}}{\tilde{\gamma}^2}  \tilde{G}_p(u)\,.
\end{equation}
Then, by using  \eqref{eq:deriv_w1}, \eqref{eq:bound-delta-l_gam-proof}, \eqref{eq:last_part}  and \eqref{eq:bound-nabla-g_gam-proof-2}, it follows that there exists a constant $\bar{C}<\infty$ such that
\begin{align*}
\Big|\frac{\partial  \ell_\gamma(\theta,u)}{\partial\theta_i}-\frac{\partial  \ell_{\tilde{\gamma}}(\theta,u)}{\partial\theta_i}\Big|&\leq  C(\gamma/\tilde{\gamma})^{\frac{d}{2}} \frac{\gamma-\tilde{\gamma}}{\tilde{\gamma}^2}  \tilde{G}_p(u)\big( \gamma^{\frac{\alpha+1}{2}}\tilde{G}_p(u)+\gamma^{1/2}\big),
\end{align*}
and the result in the last part of the lemma follows by taking $G_p(u)=C_0 \tilde{G}_p(u)^2$ for $C_0>0$ sufficiently large. The proof of the lemma is complete.
 \end{proof}
\subsection{Proof of Proposition \ref{prop:characterisation_specific}}

  Under \ref{assume:lower-bound} and \ref{assume:holder} the function $l$ is trivially lower bounded and locally integrable. Moreover, as $\theta\mapsto \ell(\theta,u)$ is lower semi-continuous by assumption, then under \ref{assume:lower-bound}  $l$ is lower semi-continuous by Proposition \ref{prop:l_lsc} (given in Section \ref{sub:l}).  Under the last assumption made in the statement of the proposition, it follows that the function $l$ is strongly lower semi-continuous. Finally, under \ref{assume:lower-bound} and as $l$ is strongly lower-semi-continuous, for any positive sequence $(\gamma_n)_{n\geq 1}$ with $\lim_{n\rightarrow\infty} \gamma_n=0$, the two sequence  $(l_{\gamma_n})_{n\geq 1}$ and $(\bar{l}_{\gamma_n})_{n\geq 1}$ epi-converge to $l$ as $n\to \infty$ by Theorem \ref{thm: epi-convergence stochastic} (Section \ref{sub:T1}) and Theorem \ref{thm:epi_conv_standard} (Section \ref{sub:T2}), respectively.

\subsubsection{Epi-convergence when \texorpdfstring{$\Ell_\gamma=\ell_\gamma$}{Lg}}\label{sub:T1}

   \begin{theorem}
\label{thm: epi-convergence stochastic}
     Let $(\gamma_n)_{n\geq 1}$ be a sequence in $(0,\infty)$ such that $\lim_{n\rightarrow\infty}\gamma_n=0$ and assume that the following conditions hold: 
       \begin{enumerate}
        \item\label{stoch lower-bound}  $\E[\inf_{\theta\in \R^d}\ell(\theta,U)] > -\infty$;

        \item\label{lsc a.e.}  the function $ \ell(\cdot,u)$ is lower semi-continuous for $\P$-a.e.~$u\in \mathsf{U}$;

        \item\label{strong lsc expected obj} for all $  \theta\in \R^d$ there exists a sequence $(\theta_n)_{n\geq 1}$ such that $\lim_{n\rightarrow\infty}\theta_n= \theta$, such that $l$ is continuous at $\theta_n$ for every $n\geq 1$, and  such that $\lim_{n\rightarrow\infty} l(\theta_n)=l(\theta)$.  
             
       \end{enumerate}
       Then, the sequence of functions $( l_{\gamma_n})_{n\geq 1}$ epi-converges to $l$.
   \end{theorem}
\begin{remark}
    By Proposition \ref{prop:l_lsc},   the function $l$  is strongly lower semi-continuous under the assumptions of the theorem.
\end{remark}
   
   \begin{proof}
     We first prove that $\liminf_{n\rightarrow\infty} l_{\gamma_n}(\theta_n)\geq l(\theta)$ for any sequence $(\theta_n)_{n\geq 1}$  converging to $\theta$. To this aim, for any $\vartheta\in \R^d$ and $u\in \mathsf{U}$ we let $h(\vartheta,u)=e^{-\ell(\vartheta,u)}$ and  $h_n(\vartheta,u)=\int_{\R^d} h(\vartheta+z, u)\phi_{\gamma_n,d}(z)\dd z$ for all $n\geq 1$, and let $(\theta_n)_{n\geq 1}$ be such that $\lim_{n\rightarrow\infty}\theta_n= \theta$. Under the assumptions of the theorem   there exists a set $U_0 \subset \mathsf{U}$ such that $\P(U_0)=1$ and such that, for any $u\in U_0$, the function $ \ell(\cdot,u)$ is lower-bounded and lower semi-continuous. We now let $u\in U_0$ be arbitrary and  note that $( \phi_{\gamma_n,d})_{n\geq 1}$ are mollifiers and that, under the assumptions of the theorem, we can apply the result given in \cite[Eq. (19)]{andrieu2024gradientfreeoptimizationintegration} with $g(\cdot)=h(\cdot,u)$ and with $g_n(\cdot)=h_n(\cdot,u)$ for all $n\geq 1$. Hence,
 \begin{equation}
 \label{eq: main chain stoch}
     \limsup_{n\rightarrow\infty} h_n(\theta_n,u)\leq \mathrm{cl}_h h(\theta,u) := \sup_{\vartheta_k\to \theta}\limsup_{k\rightarrow\infty} h(\vartheta_k,u).
 \end{equation}
 Under the assumptions on $\ell$ the function $h(\cdot,u)$ is upper semi-continuous,  implying that $\mathrm{cl}_h h(\cdot,u)=h(\cdot, u)$. Then, using \eqref{eq: main chain stoch}, we deduce that  $\limsup_{n\rightarrow\infty} h_n(\theta_n,u)\leq h(\theta,u)$  for all $u\in U_0$. To proceed further remark that    $h_n(\vartheta,u)\leq - \inf_{x\in\R^d} \ell(x,u)$ for $u \in U_0$ and all $\vartheta \in \R^d$, where $\E[-\inf_{x\in\R^d} \ell(x,U)]<\infty$ by assumption. Then, using first the reverse Fatou lemma and then the continuity and monotonicity of $\log(\cdot)$, we obtain
 \begin{align*}
     \liminf_{n\rightarrow\infty} l_{n}(\theta_n)&\geq -\E [ \limsup_{n\rightarrow\infty}  \log  h_n(\theta_n,U)]= -\E [ \log \limsup_{n\rightarrow\infty}    h_n(\theta_n,U)] 
     \geq -\E[ \log e^{-\ell(\theta,U)}]
     = l(\theta)
 \end{align*}
and thus $\liminf_{n\rightarrow\infty} l_n(\theta_n)\geq l(\theta)$ for any sequence $(\theta_n)_{n\geq 1}$ converging to $\theta$.

 To complete the proof it remains to show that for any $\theta \in \R^d$ there exists at least one sequence $(\theta_n)_{n\geq 1}$ converging to $ \theta$ which is such that $\lim_{n\rightarrow\infty} l_{\gamma_n}(\theta_n)=l(\theta)$.  To to so let   $\theta \in \R^d$ be arbitrary and note that under the assumptions of the theorem there exists a sequence $(\vartheta_k)_{k\geq 1}$ converging to $\theta$ which is such that $\lim_{k\rightarrow\infty} l(\vartheta_k)= l(\theta)$ and such that   $l$ is continuous at $\vartheta_k$ for all $k\geq 1$. 
    From this continuity, it follows from a classical result on approximation to the identity \citep[][Theorem 2.1, page 112]{stein2009real}  and from Theorem \ref{thm: bayes-laplace sandwich} that  
    \begin{equation}\label{eq:suo}
       \limsup_{n\rightarrow\infty} l_{\gamma_n}(\vartheta_k) \leq \limsup_{n\rightarrow\infty} \int_{\R^d} l(\vartheta_k) \phi_{\gamma_n,d}(x-\vartheta_k) \dd x = l(\vartheta_k),\quad\forall k\geq 1.
    \end{equation}
     Since, as proved above, we have $\liminf_{n\rightarrow\infty} l_{\gamma_n}(\vartheta_k)\leq l(\vartheta_k)$ for all $k\geq 1$, it follows from \eqref{eq:suo} that   $\lim_{n\rightarrow\infty} l_{\gamma_n}(\vartheta_k)= l(\vartheta_k)$ for all $k\geq 1$. To conclude the proof, we follow \cite[Theorem 3.7]{Ermoliev1995} and  we  define the sets $S=\crl{l(\vartheta_k), \, k\in \mathbb{N}}$ and   $S_n:=\crl{l_{\gamma_n}(\vartheta_k),\, \ k\in \mathbb{N} }$ for all $n\geq 1$. In addition, we let $\Liminf_{n} S_n$ be the set  containing the   limit of all the convergent sequences $(\alpha_n)_{n\geq 1}$ such that $\alpha_n \in S_n$ for all $n\geq 1$. 
     As shown above,   $S \subset \Liminf_{n} S_n$ and thus, since the set $\Liminf_{n} S_n$ is closed,  it follows that $S \subset \text{cl}(S) \subset \Liminf_{n} S_n$. Moreover, since $(\vartheta_k)_{k\geq 1}$ is chosen to be such that   $l(\theta)\in \text{cl}(S)$, it follows that there exists a sequence $(\vartheta_{k_n})_{n\geq 1}$ such that $\lim_{n\rightarrow\infty} l_{\gamma_n}(\vartheta_{k_n}) = l(\theta)$. The proof of the theorem is complete.
   \end{proof}

\subsubsection{Epi-convergence when \texorpdfstring{$\Ell_\gamma=\bar{\ell}_\gamma$}{Lg}}\label{sub:T2}

\begin{theorem}
\label{thm:epi_conv_standard}
Let $(\gamma_n)_{n\geq 1}$ be a sequence in $(0,\infty)$ such that $\lim_{n\rightarrow\infty}\gamma_n=0$ and assume that the function $l$ is strongly lower semi-continuous and that $\E[\inf_{\theta\in \R^d}\ell(\theta,U)] > -\infty$. Then, the sequence of functions $(\bar{l}_{\gamma_n})_{n\geq 1}$ epi-converges to $l$.
 
\end{theorem}
\begin{proof}
For all $u\in\mathsf{U}$ let $\eta(u)=\inf_{\theta\in \R^d}\ell(\theta,u)$  and note that $\E[\eta(U)]>-\infty$ by assumption. Let $\theta\in\R^d$ and $(\theta_n)_{n\geq 1}$ be a sequence converging to $\theta$. Then, for all $n\geq 1$ we have
\begin{equation}\label{eq:epi1}
\begin{split}
\bar{l}_{\gamma_n}(\theta_n)&=\E\bigg[\int_{\R^d} \ell(\theta_n+\gamma_n^{1/2}z,U)\phi_{d,1}(z)\dd z\bigg]\\
&=\E[\eta(U)]+\E\bigg[\int_{\R^d} \Big(\ell(\theta_n+\gamma_n^{1/2}z,U)-\eta(U)\Big)\phi_{d,1}(z)\dd z\bigg]\\
&=\E[\eta(U)]+ \int_{\R^d} \Big(l(\theta_n+\gamma_n^{1/2}z)-\E[\eta(U)]\Big)\phi_{d,1}(z)\dd z
\end{split}
\end{equation}
where the last inequality holds by Tonelli's theorem. By Fatou's lemma, we have
\begin{equation}\label{eq:epi2}
\begin{split}
\liminf_{n\rightarrow\infty}\int_{\R^d} \Big(l(\theta_n+\gamma_n^{1/2})-\E[\eta(U)]\Big)\phi_{d,1}(z)\dd z&\geq \int_{\R^d}\liminf_{n\rightarrow\infty}\Big(l(\theta_n+\gamma_n^{1/2}z)-\E[\eta(U)]\Big)\phi_{d,1}(z)\\
\end{split}
\end{equation}
where the second inequality uses the fact that, since $l$ is assumed to be lower semi-continuous, we have $\liminf_{n\rightarrow\infty} l(\theta_n+\gamma_n^{1/2}z)\geq l(\theta)$ for all $z\in\R^d$. By combining \eqref{eq:epi1} and \eqref{eq:epi2}, it follows that $\liminf_{n\rightarrow\infty}\bar{l}_{\gamma_n}(\theta_n)\geq l(\theta)$.

 To complete the proof it remains to show that for any $\theta \in \R^d$ there exists at least one sequence $(\theta_n)_{n\geq 1}$ converging to $ \theta$ which is such that $\lim_{n\rightarrow\infty} \bar{l}_{\gamma_n}(\theta_n)=l(\theta)$. To do so let $\theta\in\R^d$ and note that since, by assumption, the funcion $l$ is strongly lower semi-continuous   there exists a sequence $(\vartheta_k)_{k\geq 1}$ converging to $\theta$ such that  $l$ continuous at $\vartheta_k$ for all $k\geq 1$ and such that $\lim_{k\rightarrow\infty}l(\vartheta_k)=l(\theta)$. Then, using a classical result on approximation to the identity \citep[][Theorem 2.1, page 112]{stein2009real}, for any $k\geq 1$ we have $\limsup_{n\rightarrow\infty} \bar{l}_n(\vartheta_k)=l(\vartheta_k)$ and we can now proceed  as in the very last part proof of Theorem \ref{thm: epi-convergence stochastic} to show that  there exists a sequence $(\vartheta_{k_n})_{n\geq 1}$ such that $\lim_{n\rightarrow\infty} \bar{l}_{\gamma_n}(\vartheta_{k_n}) = l(\theta)$. The proof is complete.
\end{proof}

\subsubsection{Lower semi-continuity of  \texorpdfstring{$l$}{Lg}}\label{sub:l}

\begin{proposition}
\label{prop:l_lsc}
Assume that $\E[\inf_{\theta\in \R^d}\ell(\theta,U)] > -\infty$ and that for $\P$-a.e. $u\in\mathsf{U}$   the function $\ell(\cdot,u)$
 is lower semi-continuous. Then, the function $l$ is  lower semi-continuous.
\end{proposition}
\begin{proof}
 The result of the proposition directly follows from applying \eqref{eq:epi1}-\eqref{eq:epi2} with $\gamma_n=0$ for all $n \geq 1$.
\end{proof}

\subsection{Proof of Proposition \ref{prop:lsc-unconstrained-example}\label{p-propExample}}
\begin{proof}
Given the particular definition of the functions $l$ and $\ell$ that we consider, to prove the proposition it is enough to show that for any $u\in\mathsf{U}$ the function $\theta\mapsto\ell(\theta,u)$ is strictly lower semi-continuous. The result for $l$ is a particular case obtained for $n_{\mathrm{batch}}=n_{\mathrm{data}}$ and a particular $u$.

To this aim, let $u\in \mathsf{U}$ and $I_u=\{ (i_{u,k},j_{u,k}),\,k=1,\dots,n_{\mathrm{batch}}\}$. As preliminary calculations let $\theta\in\mathbb{R}^{p-1}$ be such
that $\sigma^{-1}(\theta)^{\top}(z_{i}-z_{j})\neq0$ for some $(i,j)\in I_u$, and note that the mapping 
$\theta'\mapsto\mathbf{1}\{\sigma^{-1}(\theta')^{\top}(z_{i}-z_{j})<0\}$ is
constant in a sufficiently small neighbourhood of $\theta$, by continuity
of linear mappings and $\sigma^{-1}$. Consequently $\theta$ is a point of discontinuity
for the function $\ell(\cdot, u)$ if and only if $\sigma^{-1}(\theta)^{\top}(z_{i}-z_{j})=0$ for some  $(i,j)\in I_u$, that is if and only if $D_{u,\theta}:=I_u\cap D_{\sigma^{-1}(\theta)}\neq\emptyset$.

We now let $\theta\in\mathbb{R}^{p-1}$ be fixed,   $(\theta_n)_{n\geq 1}$ is a sequence in $\R^{p-1}$ such that $\lim_{n\rightarrow\infty}\theta_{n}=\theta$, and to simplify notation we let $C= 2n_{+}(n_{\mathrm{data}}-n_+)/(n_\mathrm{data}(n_\mathrm{data}-1)n_{\mathrm{batch}})$. Then, using the convention that empty sums equal zero, we can write $\ell(\theta,u)$ as follows
\begin{align}\label{eq:cut}
\ell(\theta,u)=C\sum_{(i,j)\in D_{u,\theta}}\mathbf{1}\{\sigma^{-1}(\theta_{n})^{\top}(z_{i}-z_{j})<0\}+C\sum_{(i,j)\in D_{u,\theta}^{\complement}}\mathbf{1}\{\sigma^{-1}(\theta_{n})^{\top}(z_{i}-z_{j})<0\}.
\end{align}
Assume   that the set $D_{u,\theta}^{\complement}$ is non-empty. In this case, by the continuity of the mapping $\sigma^{-1}$, there exists an $n_{0}=n_{0}(\theta)\in\mathbb{N}$ such that
$\mathbf{1}\{\sigma^{-1}(\theta_{n})^{\top}(z_{i}-z_{j})<0\}=\mathbf{1}\{\sigma^{-1}(\theta)^{\top}(z_{i}-z_{j})<0\}$ for all $n\geq n_{0}$ and all $(i,j)\in D_{\theta}^{\complement}$, implying that
\begin{equation}
\lim_{n\rightarrow\infty}\sum_{(i,j)\in D_{u,\theta}^{\complement}}\mathbf{1}\{\sigma^{-1}(\theta_{n})^{\top}(z_{i}-z_{j})<0\}=\sum_{(i,j)\in D_{u,\theta}^{\complement}}\mathbf{1}\{\sigma(\theta)^{\top}(z_{i}-z_{j})<0\}\,.\label{eq:convergence-continuous-component}
\end{equation}
Assume now that the set $D_{u,\theta}$ is non-empty. Then,  for any $(i,j)\in D_{u,\theta}$ we have $\mathbf{1}\{\sigma^{-1}(\theta)^{\top}(z_{i}-z_{j})<0\}=0$, and thus
\begin{equation}\label{eq:convergence-noncontinuous-component}
\liminf_{n\rightarrow\infty}\mathbf{1}\{\sigma^{-1}(\theta_{n})^{\top}(z_{i}-z_{j})<0\}\geq 0=\mathbf{1}\{\sigma^{-1}(\theta)^{\top}(z_{i}-z_{j})<0\},\quad\forall (i,j)\in D_{u,\theta}.
\end{equation}
By using \eqref{eq:cut}-\eqref{eq:convergence-noncontinuous-component} and the superadditivity of the limit inferior, we deduce  that $\liminf_{n}\ell(\theta_{n},u)\geq \ell(\theta,u)$, showing   that the function $\theta'\mapsto \ell(\theta',u)$ is lower semi-continuous at $\theta$.

We now establish that this function is actually strongly lower semi-continuous at $\theta$. To this end assume first that  $D_{u,\theta}=\emptyset$. In this case, by applying  \eqref{eq:convergence-continuous-component} with $D_{u,\theta}^{\complement}=I_u$, it follows from \eqref{eq:cut} that $\lim_{n\rightarrow\infty}\ell(\theta_n,u)=\ell(\theta,u)$ for any sequence $(\theta_n)_{n\geq 1}$ converging to $\theta$. Together with the fact that, as shown above, the function $\ell(\cdot, u)$ is continuous at $\theta$ when   $D_{u,\theta}=\emptyset$, we can easily conclude that   this function is strictly lower semi-continuous at $\theta$ if $D_{u,\theta}=\emptyset$.

Assume now that $D_{u,\theta}\neq\emptyset$ and let   $(i_{0},j_{0})\in D_{u,\theta}$. In this scenario, for any $n\geq1$ we let $\vartheta_{n}=\sigma^{-1}(\theta)+n^{-1}(z_{i_{0}}-z_{j_{0}})/\|z_{i_{0}}-z_{j_{0}}\|$. Then, since by assumption we have $(z_{i_{1}}-z_{j_{1}})^{\top}(z_{i_{2}}-z_{j_{2}})\geq 0 $ for all $(i_1,j_1), (i_2,j_2)\in D_{u,\theta}$,  it follows from the definition of $(\vartheta_n)_{n\geq 1}$ that $\mathbf{1}\{\vartheta_{n}^{\top}(z_{i}-z_{j})<0\}=0$ for all $(i,j)\in D_{u,\theta}$ and $n\geq 1$.  Consequently, letting $\theta_n=\sigma(\vartheta_n)$ for all $n\geq 1$, we have
\begin{align*}
\lim_{n\rightarrow\infty}\mathbf{1}\{\sigma^{-1}(\theta_{n})^{\top}(z_{i}-z_{j})<0\}=\mathbf{1}\{\sigma^{-1}(\theta)^{\top}(z_{i}-z_{j})<0\},\quad\forall (i,j)\in D_{u,\theta}\neq \emptyset
\end{align*}
which, together with \eqref{eq:cut} and \eqref{eq:convergence-continuous-component}, shows that   $\lim_{n\rightarrow}\ell(\theta_{n})=\ell(\theta)$. Since the function $\sigma$ is continuous and $\lim_{n\rightarrow\infty}\vartheta_n=\sigma^{-1}(\theta)$, it follows that $\lim_{n\rightarrow\infty}\theta_n=\theta$. Finally, to show that $\theta_n$ is, for $n$ large enough, a point of continuity of the function $\ell(\cdot,u)$, remark first that
 $\vartheta_{n}^{\top}(z_{i}-z_{j})\neq0$ for all $(i,j)\in D_{u,\theta}$. On the other hand, since  $\lim_{n\rightarrow\infty}\vartheta_n=\sigma^{-1}(\theta)$, for all  $(i,j)\in D_{u,\theta}^\complement$ there exists an $n_{i,j}\in\mathbb{N}$ such that  $\vartheta_{n}^{\top}(z_{i}-z_{j})\neq0$ for all $n\geq n_{i,j}$. Noting that  $\sharp I_u<\infty$, it follows that there exists an $n'\in\mathbb{N}$ such that  $\vartheta_{n}^{\top}(z_{i}-z_{j})\neq0$ for all $(i,j)\in I_u$ and all $n\geq n'$. As proved above, this implies that $\theta_n$ is a point of continuity of the function $\ell(\cdot,u)$, which shows that this function is strongly lower semi-continuous at $\theta$. Since the pair  $(\theta,u)\in\R^{p-1}\mathsf{U}$ above is arbitrary,  the result of the proposition follows.
\end{proof}

\appendix

\section{Functional minimisation by coordinate descent} \label{app:coordinate-descent}

From the Laplace principle (see e.g. \cite{andrieu2024gradientfreeoptimizationintegration}) for $\tilde{\pi}\propto \pi_{\theta,\gamma} \exp(-l)$ we know that
\begin{equation}
			\label{eq: laplace principle}
			\nu \mapsto \mathrm{KL}(\nu, \tilde{\pi})=\int l(x)\nu(x)\mathrm{d}x+\mathrm{KL}(\nu,\pi_{\theta,\gamma})+\log\int e^{-l(x)}\pi_{\theta,\gamma}(x)\mathrm{d}x.
		\end{equation}
is minimised at $\nu = \tilde{\pi}$ (Bayes' rule) and $\mathrm{KL}(\nu, \tilde{\pi})=0$ implying that at this value of $\nu$ we have $\Phi_1(\nu,\theta;\gamma)= l_\gamma(\theta)$.

 \section{The Bayes-Laplace Sandwich Theorem}\label{app:sandwich}

    \begin{theorem}[Bayes-Laplace Sandwich]
    \label{thm: bayes-laplace sandwich}
       Let $g:\R^d\to \R$ be a locally integrable function. Fix $\theta\in \R^d$ and let $g_{\psi}(\theta)=-\log\int_{\R^d} e^{-g(\theta+z)}\psi(z)\mathrm{d}z$, where $\psi:\R^d\to \R_+$ satisfies $\int_{\R^d} \psi(z)\mathrm{d}z=1$. Then,
        \begin{equation*}
            \int_{\R^d} g(\theta+z)\frac{e^{-g(\theta+z)}\psi(z)}{\int_{\R^d} e^{-g(\theta+z')}\psi(z')\mathrm{d}z'}\mathrm{d}z\leq g_{\psi}(\theta)\leq \int_{\R^d} g(\theta+z)\psi(z)\mathrm{d}z.
        \end{equation*}
    \end{theorem}
    \begin{proof}
    Fix $\theta\in \R^d$. Then, by using Laplace Principle \eqref{eq: laplace principle}, we have
    \begin{align*}
                g_{\psi}(\theta)&=\int_{\R^d} g(\theta+z)\frac{e^{-g(\theta+z)}\psi(z)}{\int_{\R^d} e^{-g(\theta+z')}\psi(z')\mathrm{d}z'}\mathrm{d}z+\mathrm{KL}\round{\frac{e^{-g(\theta+\cdot)}\psi}{\psi(e^{-g(\theta+\cdot)})}, \psi} \nonumber \\
                &\geq \int_{\R^d} g(\theta+z)\frac{e^{-g(\theta+z)}\psi(z)}{\int_{\R^d} e^{-g(\theta+z')}\psi(z')\mathrm{d}z'}\mathrm{d}z
            \end{align*}
            where the inequality uses the fact that $\mathrm{KL}(\cdot,\cdot)\geq 0$.
     On the other hand, by Jensen's Inequality, we have
     \begin{align*}
     g_{\psi}(\theta)=-\log\int_{\R^d} e^{-g(\theta+z)}\psi(z)\mathrm{d}z\leq \int_{\R^d} g(\theta+z)\psi(z)\mathrm{d}z.
    \end{align*}
    and the proof of the theorem is complete.
    
    \end{proof}

\section{Technical results for proving Propositions \ref{prop:2point}-\ref{prop:Bayes}} \label{subsec:from-gen-to-example}

\subsection{Some technical lemmas}
\begin{lemma}
There exists a constant $C<\infty$ such that for any $0<\tilde{\gamma}\leq\gamma<\infty$
and $z\in\mathbb{R}^{d}$,
\begin{align*}
\left|\phi_{d,\gamma}(z)-\phi_{d,\tilde{\gamma}}(z)\right| & \leq C\left(\frac{\gamma}{\tilde{\gamma}}\right)^{d/2}\frac{\gamma-\tilde{\gamma}}{\tilde{\gamma}}\left[1+\gamma^{-1}\|z\|^{2}\right]\phi_{d,\gamma}(z)\,,\\
\left|\gamma^{-1}\phi_{d,\gamma}(z)-\tilde{\gamma}^{-1}\phi_{d,\tilde{\gamma}}(z)\right| & \leq C\left(\frac{\gamma}{\tilde{\gamma}}\right)^{d/2}\frac{\gamma-\tilde{\gamma}}{\tilde{\gamma}^{2}}\left[1+\gamma^{-1}\|z\|^{2}\right]\phi_{d,\gamma}(z)\,.
\end{align*}
\end{lemma}

\begin{proof}
Let  $0\leq\tilde{\gamma}\leq\gamma$ and $z\in\R^d$ be arbitrary, and let $c=(\gamma/\tilde{\gamma})^{d/2}\geq1$. 

To prove the first part of the lemma assume first that $\phi_{d,\gamma}(z)-\phi_{d,\tilde{\gamma}}(z)\geq 0$, so that
\begin{align*}
\left|\phi_{d,\gamma}(z)-\phi_{d,\tilde{\gamma}}(z)\right| & =\left(1-c\exp\left(-\frac{\gamma-\tilde{\gamma}}{2\gamma\tilde{\gamma}}\|z\|^{2}\right)\right)\phi_{d,\gamma}(z)\\
 & \leq c\left(1-\exp\left(-\frac{\gamma-\tilde{\gamma}}{2\gamma\tilde{\gamma}}\|z\|^{2}\right)\right)\phi_{d,\gamma}(z)\\
 & \leq c\frac{\gamma-\tilde{\gamma}}{2\gamma\tilde{\gamma}}\|z\|^{2}\phi_{d,\gamma}(z)
\end{align*}
where the last inequality uses the fact that $\exp(-x)\geq1-x$ for all $x\in\mathbb{R}$.

Assume now that $\phi_{d,\gamma}(z)-\phi_{d,\tilde{\gamma}}(z)< 0$. Then,
\begin{align*}
\left|\phi_{d,\gamma}(z)-\phi_{d,\tilde{\gamma}}(z)\right| =\left(\frac{\phi_{d,\tilde{\gamma}}(z)}{\phi_{d,\gamma}(z)}-1\right)\phi_{d,\gamma}(z) \leq(c-1)\phi_{d,\gamma}(z) \leq \frac{cd}{2}\frac{\gamma-\tilde{\gamma}}{\tilde{\gamma}}\phi_{d,\gamma}(z) 
\end{align*}
where the last inequality uses the fact that since
\[
(1+x)^{k/2}-1=k/2\int_{0}^{x}(1+u)^{k/2-1}{\rm d}u\leq k/2(1+x)^{k/2}x,\quad\forall k\in\mathbb{N},\quad\forall x\in (0,\infty),
\]
we have
\begin{align}
\Big(1+\frac{\gamma-\tilde{\gamma}}{\tilde{\gamma}}\Big)^{k/2}-1\leq \frac{k}{2}\,(\gamma/\tilde{\gamma})^{k/2}\,\frac{\gamma-\tilde{\gamma}}{\tilde{\gamma}},\quad\forall k\in\mathbb{N}.\label{eq:int}
\end{align}
Consequently
\[
\left|\phi_{d,\gamma}(z)-\phi_{d,\tilde{\gamma}}(z)\right|\leq \frac{cd}{2}\frac{\gamma-\tilde{\gamma}}{\tilde{\gamma}}\left[1+\gamma^{-1}\|z\|^{2}\right]\phi_{d,\gamma}(z) 
\]
showing the first result of the lemma.

Similarly,  to show the second part of the lemma assume first that $\gamma^{-1}\phi_{d,\gamma}(z)-\tilde{\gamma}^{-1}\phi_{d,\tilde{\gamma}}(z)\geq 0$. Then,
\begin{align*}
|\gamma^{-1}\phi_{d,\gamma}(z)-\tilde{\gamma}^{-1}\phi_{d,\tilde{\gamma}}(z)|& =\tilde{\gamma}^{-1}\left[\frac{\tilde{\gamma}}{\gamma}-c\exp\left(-\frac{\gamma-\tilde{\gamma}}{2\gamma\tilde{\gamma}}\|z\|^{2}\right)\right]\phi_{d,\gamma}(z)\\
 & \leq c\tilde{\gamma}^{-1}\left[1-\exp\left(-\frac{\gamma-\tilde{\gamma}}{2\gamma\tilde{\gamma}}\|z\|^{2}\right)\right]\phi_{d,\gamma}(z)\\
 & \leq c\frac{\gamma-\tilde{\gamma}}{2\gamma\tilde{\gamma}^{2}}\|z\|^{2}\phi_{d,\gamma}(z)
\end{align*}
while, if $\gamma^{-1}\phi_{d,\gamma}(z)-\tilde{\gamma}^{-1}\phi_{d,\tilde{\gamma}}(z)< 0$,
\begin{align*} 
|\tilde{\gamma}^{-1}\phi_{d,\tilde{\gamma}}(z)-\gamma^{-1}\phi_{d,\gamma}(z)|& =\gamma^{-1}\left[\frac{\gamma}{\tilde{\gamma}}\frac{\phi_{d,\tilde{\gamma}}(z)}{\phi_{d,\gamma}(z)}-1\right]\phi_{d,\gamma}(z)\\
& \leq\gamma^{-1}\left[\frac{\gamma}{\tilde{\gamma}}c-1\right]\phi_{d,\gamma}(z)\\
 & \leq\gamma^{-1}\frac{c(d+1)}{2}\frac{\gamma-\tilde{\gamma}}{\tilde{\gamma}}\phi_{d,\gamma}(z)\\
 & \leq \frac{c(d+1)}{2}\frac{\gamma-\tilde{\gamma}}{\tilde{\gamma}^{2}}\phi_{d,\gamma}(z)\,,
\end{align*}
where the third inequality uses \eqref{eq:int}. Consequently
\[
\left|\gamma^{-1}\phi_{d,\gamma}(z)-\tilde{\gamma}^{-1}\phi_{d,\tilde{\gamma}}(z)\right|\leq  \frac{c(d+1)}{2}\frac{\gamma-\tilde{\gamma}}{\tilde{\gamma}^{2}}\left[\gamma^{-1}\|z\|^{2}+1\right]\phi_{d,\gamma}(z)
\]
and the proof of the lemma is complete.
\end{proof}
\begin{corollary} \label{cor:delta-phi-gam}
There exists a constant $C<\infty$  such that, for any function $f\colon \mathbb{R}^d \rightarrow \mathbb{R}$
and constants $0<\tilde{\gamma}\leq\gamma<\infty$,
\begin{align*}
\left|\int_{\mathbb{R}^{d}}f(z)\Big(\phi_{d,\gamma}( z)-\phi_{d,\tilde{\gamma}}( z)\Big)\dd z\right| & \leq C\left(\frac{\gamma}{\tilde{\gamma}}\right)^{d/2}\frac{\gamma-\tilde{\gamma}}{\tilde{\gamma}}\int|f(\gamma^{1/2}z)|(1+\|z\|^{2})\phi_{d,1}( z)\dd z
\end{align*}
\[
\left|\int_{\mathbb{R}^{d}}f(z)\Big(\gamma^{-1}\phi_{d,\gamma}( z)-\tilde{\gamma}^{-1}\phi_{d,\tilde{\gamma}}( z)\Big)\dd z\right|\leq C\left(\frac{\gamma}{\tilde{\gamma}}\right)^{d/2}\frac{\gamma-\tilde{\gamma}}{\tilde{\gamma}^{2}}\int |f(\gamma^{1/2}z)|(1+\|z\|^{2})\phi_{d,1}(z)\dd z.
\]
\end{corollary}

\begin{lemma}\label{lemma:cov}
Let $X$ be a  $[0,\infty)$-valued random variable,   $ f:\R\rightarrow\R$ be a non-decreasing   function and $g:\R\rightarrow\R$ be a non-increasing     function such that $g(x)> 0$ for any $x \in \R$. Then,
\begin{align*}
\frac{\E[f(X)]}{\E[g(X)]}\leq \E\bigg[\frac{f(X)}{g(X)}\bigg].
\end{align*}
\end{lemma}
\begin{proof}
We have
\begin{align*}
0\leq \mathrm{Cov}\Big(\frac{1}{g(X)}, f(X)\Big)=\E\bigg[\frac{f(X)}{g(X)}\bigg]-\E\big[1/g(X)\big]\E[f(X)]\leq \E\bigg[\frac{f(X)}{g(X)}\bigg]-\frac{\E[f(X)]}{\E[g(X)]}
\end{align*}
where the first inequality follows from the well-known covariance inequality \cite{SCHMIDT201491}   while the last inequality holds by Jensen's inequality.
\end{proof}

\subsection{Additional technical lemmas under \ref{assume:holder2}} \label{subsec:tech-res-A3}

\begin{lemma}\label{lemma:J}
Assume that,  for some $\upsilon_1\in (0,1]$ and function  $J_1:\mathsf{U}\to \R_+$,  Assumption \ref{assume:holder2} holds with $\upsilon=\upsilon_1$ and $J=J_1$. Then  Assumption \ref{assume:holder2} also holds with $\upsilon=\upsilon_2:=\upsilon_1/2$ and $J=J_2:=J_1+1$.
\end{lemma}
\begin{proof}
For all $u\in\mathsf{U}$ we have, using Cauchy-Schwartz's inequality,
\begin{align*}
\int_{\R}e^{\upsilon_2 J_2(u)(1+|z|^\beta)}\phi_{1,1}(\dd z)&=\int_\R e^{\upsilon_1/2  (1+|z|^\beta)}e^{\upsilon_1/2 J_1(u) (1+|z|^\beta)}\phi_{1,1}(z)\dd z\\
&\leq \Big(\int_\R e^{\upsilon_1  (1+|z|^\beta)}\phi_{1,1}(z)\dd z\Big)^{1/2}\Big(\int_\R e^{\upsilon_1 J_1(u) (1+|z|^\beta)}\phi_{1,1}(z)\dd z\Big)^{1/2}
\end{align*}
and thus, using Jensen's inequality,
\begin{align*}
\E\Big[\int_{\R } e^{\upsilon_2 J_2(U)(1+|z|^\beta)}\phi_{1,1}(z)\dd z\Big]^2\leq  \Big(\int_\R e^{\upsilon_1   (1+|z|^\beta)}\phi_{1,1}(z)\dd z\Big) \Big(\E\Big[\int_\R e^{\upsilon_1 J_1(U) (1+|z|^\beta)}\phi_{1,1}(z)\dd z\Big]\Big) <\infty.
\end{align*}
The result of the lemma follows.
\end{proof}

\begin{lemma}\label{lemma:assume}
Assume that Assumption \ref{assume:holder2} holds.  Then, for all  $p\in[1,\infty)$ there exists a constant $\bar{c}_p\in(0,1]$ such that,  with  $\alpha$, $\beta$ and $J$  as \ref{assume:holder2},
\begin{align*}
\E\bigg[ \bigg(J(U)\int_{\R^d} e^{\bar{c}_p J(U) (1+\|z\|^{\beta})} (1+\|z\|^4)\phi_{d,1}(z)\dd z\bigg)^p\bigg]<\infty.
\end{align*}
\end{lemma}
\begin{proof}

We start with some preliminary calculations. Let $c>0$, $u\in\mathsf{U}$  and $b\in [1,2]$. Then, noting that $x^4\leq (4/(c b))^{1/b} e^{c x^b}$ for all $x\in[0,\infty)$ and  that   $\big(\sum_{i=1}^d|x_i|\big)^b\leq d^{\beta-1}\sum_{i=1}^d |x_i|^b$ for all $x\in\R^d$ by Jensen's inequality, and recalling that  $\|x\|\leq \sum_{i=1}^d|x_i|$ for all $x\in\R^d$,  we have 
\begin{align*}
G_{b,c}(u) :&=\int_{\R^d} e^{c J(u) (2+\|z\|^{b})} (1+\|z\|^4)\phi_{d,1}(z)\dd z \\
&\leq e^{2c J(u)}\int_{\R^d} e^{c J(u)  \|z\|^{b}} \phi_{d,1}(z)\dd z+\Big(\frac{4}{c b}\Big)^{1/b}e^{2c J(U)}\int_{\R^d} e^{c (J(u)+1)  \|z\|^{b}} \phi_{d,1}(z)\dd z\\
&\leq \bigg(1+\Big( \frac{4}{c b}\Big)^{1/b}\bigg)e^{2c J(u)}\Big(\int_{\R} e^{ cd^{b-1} (J(u)+1)   |z|^{b}} \phi_{1,1}(z)\dd z\Big)^d\\
&\leq \bigg(1+\Big( \frac{4}{c b}\Big)^{1/b}\bigg)e^{2c J(u)} \int_{\R} e^{c d^b (J(u)+1)   |z|^{b}} \phi_{1,1}(z)\dd z \\
&\leq \bigg(1+\Big( \frac{4}{cb}\Big)^{1/b}\bigg)  \int_{\R} e^{2cd^b (J(u)+1)(1 + |z|^{b})} \phi_{1,1}(z)\dd z
\end{align*}
where the third inequality holds by Jensen's inequality. By Lemma~\ref{lemma:J} we can without loss of generality assume that $J(\tilde{u})\geq 1$ for all $\tilde{u}\in\mathsf{U}$, and thus for all $p\in\mathbb{N}$ we have, using Jensen's inequality,
\begin{align}\label{eq:key_G}
J(u)^pG^p_{b,c}(u)&\leq \bigg(1+\Big( \frac{4}{c b}\Big)^{1/b}\bigg)^p J(u)^p  \int_{\R} e^{4cpd^{b}  J(u) (1 + |z|^{b})} \phi_{1,1}(z)\dd z\,.
\end{align}
We now use \ref{assume:holder2} to prove that $\E\big[J(U)^pG^p_{\beta,c}(U)\big]<\infty$ for $c>0$ sufficiently small, from which the lemma follows. If $\beta\geq 1$ then the result of the lemma can be readily obtained by using \eqref{eq:key_G} with $b=\beta$ and $c>0$ sufficiently small. If $\beta<1$ then for all $c>0$ and $u\in\mathsf{U}$ we have 
\begin{align*}
 \int_{\R^d} e^{c J(U) (1+\|z\|^{\beta})} (1+\|z\|^4)\phi_{d,1}(\dd z)\leq  \int_{\R^d} e^{c J(U) (2+\|z\|)} (1+\|z\|^4)\phi_{d,1}(\dd z)=G_{1,c}(u),
\end{align*}
in which case the  result of the lemma can be obtained by using \eqref{eq:key_G} with $b=1$ and $c>0$ sufficiently small. The proof of the lemma is complete.
\end{proof}

 \begin{lemma}\label{lemma:zero_mean}
 Assume that \ref{assume:holder2} holds and let $f:\R^d\rightarrow\R$ be such that $\int_{\R^d}f(z)\phi_{d,1}(\dd x)=0$ and $g \colon [0,\infty) \rightarrow [0,\infty)$ be a non-decreasing function. Then, for all $\gamma\in(0,1]$ and all $(\theta,u)\in\R^d\times\mathsf{U}$, and  with $\alpha$, $\beta$ and $J$ as in \ref{assume:holder2},
\begin{enumerate}
    \item  we have
    \begin{align*}
 \bigg|\int_{\R^d} e^{\ell(\theta,u)-\ell(\theta+\gamma^{1/2}z,u)}f(z)\phi_{d,1}(\dd z)\bigg| \leq 2\gamma^{\frac{\alpha}{2}}J(u)\int_{\R^d}(1+\|z\|^\beta) e^{2\gamma^{\frac{\alpha}{2}}J(u)(1+\|z\|^\beta)}|f(z)|\phi_{d,1}(z)\dd z.
 \end{align*}
\item for  any $  \tilde{\gamma}\in(0, \gamma]$, we have
    \begin{align*}
     \frac{\int_{\R^d}  e^{ -\ell(\theta+\gamma^{1/2}z,u)}g(\|z\|)   \phi_{d,1}(\dd z)}{\int_{\R^d}  e^{-\ell(\theta+\tilde{\gamma}^{1/2}z,u)} \phi_{d,1}(\dd z)}
      \leq \int_{\R^d}g(\|z\|) e^{4\gamma^{\frac{\alpha}{2}}J(u)(1+\|z\|^\beta)}\phi_{d,1}(\dd z)\,,
 \end{align*}
 \item assuming that $|f(z)| \leq g(\|z\|)$ for all $z\in\R^d$, we have
 \begin{align*}
     \frac{\big|\int_{\R^d} e^{-\ell(\theta+\gamma^{1/2}z,u)}f(z)\phi_{d,1}( z)\dd z\big|}{\int_{\R^d} e^{-\ell(\theta+\gamma^{1/2}z,u)} \phi_{d,1}(  z)\dd z}
      \leq 2\gamma^{\frac{\alpha}{2}}J(u)\int_{\R^d}(2+\|z\|^2)g(\|z\|) e^{4\gamma^{\frac{\alpha}{2}}J(u)(1+\|z\|^\beta)}\phi_{d,1}(  z)\dd z.
 \end{align*}
 \end{enumerate}
 \end{lemma}
 \begin{proof}
 Let   $J$, $f$, $g$ and $\alpha,\beta$ be as in the statement of the lemma, and let $\gamma\in (0,1]$ and $(\theta,u)\in\R^d\times\mathsf{U}$ be arbitrary
 
 Since $\gamma \leq 1$ and $\alpha\leq \beta\leq 2$  we have, under \ref{assume:holder2} and for all $z\in\R^d$,
\begin{equation}\label{eq:ll}
\begin{split}
    |\ell(\theta+\gamma^{1/2}z,u)-\ell(\theta,u)|&\leq \gamma^{\alpha/2}J(u)\big(\|z\|^\alpha+\|z\|^\beta\big)\\
    & \leq 2\gamma^{\alpha/2}J(u)\big(1+\|z\|^\beta\big)
\end{split}
\end{equation}
and thus, using the fact that $|1-e^x|\leq |x|e^{|x|}$ for all $x\in\R$,
\begin{equation*}
\begin{split}
\bigg|\int_{\R^d} e^{\ell(\theta,u)-\ell(\theta+\gamma^{1/2}z,u)}f(z)\phi_{d,1}( z)\dd z\bigg|&=\bigg|\int_{\R^d} \Big(e^{\ell(\theta,u)-\ell(\theta+\gamma^{1/2}z,u)}-1\Big)f(z)\phi_{d,1}( z)\dd z\bigg|\\
&\leq 2J(u)\gamma^{\alpha/2}\int_{\R^d}(1+\|z\|^\beta) e^{2\gamma^{\alpha/2}J(u)(1+\|z\|^\beta)}|f(z)|\phi_{d,1}(z)\dd z
\end{split}
\end{equation*}
showing the first part of the lemma. To show the second part of the lemma remark that, under \ref{assume:holder2} and since $\tilde{\gamma}\leq \gamma$,  we have
\begin{align}\label{eq:denom}
\int_{\R^d} e^{\ell(\theta,u)-\ell(\theta+\tilde{\gamma}^{1/2}z,u)} \phi_{d,1}( z)\dd z\geq \int_{\R^d} e^{-2\gamma^{\alpha/2}J(u)(1+\|z\|^\beta)}\phi_{d,1}( z)\dd z  
\end{align}
and the  second part of the lemma then follows from \eqref{eq:ll} and Lemma~\ref{lemma:cov}. Finally, the last part of the lemma follows from first  part  of the lemma, \eqref{eq:ll}, Lemma~\ref{lemma:cov} and $1+\|z\|^\beta\leq 2+\|z\|^2$.

 \end{proof}

\section{Epi-convergence and lower-semicontinuity}
\label{app:epi_conv_lsc}
\begin{definition}[Lower semi-continuity]{\cite[Def. 1.5]{RockWets98}}
\label{def: lsc}
A function $f:\R^d\to \R$ is said to be 
\begin{enumerate}
  \item lower semi-continuous ($\mathrm{lsc}$) at $\theta_{0}\in \R^d$ if 
  \begin{equation}
  \label{eq: lsc}\liminf_{\theta\to \theta_{0}}f(\theta)\geq f(\theta_{0}).
  \end{equation}
  \item lower semi-continuous if the above holds for any $\theta_{0}\in \R^d$.
\end{enumerate}   
\end{definition}

By definition, any lower semi-continuous, lower-bounded function has a minimum on $\R^d$. 
    
\begin{definition}[Strong lower semi-continuity]
\label{def: slsc}
   A function $f:\R^d\to \R$ is said to be
   \begin{enumerate}
       \item strongly lower semi-continuous ($slsc$) at $\theta\in \R^d$
         if it is lower semi-continuous at $\theta\in \R^d$ and there
         exists a sequence $\crl{\theta_n, \ n\in \mathbb{N}}$,
         $\theta_n\to \theta$, with $f$ continuous at every $\theta_n$, and $f(\theta_n)\to f(\theta)$.
       \item strongly lower semi-continuous if the above holds for any $\theta\in \R^d$.
   \end{enumerate}
    
\end{definition}

\begin{definition}[Epi-convergence] \label{def-epilimits-1}
  A sequence of functions $\crl{f_n:\R^d\to \R, \ n\in\mathbb{N}}$
  epi-converges to a function $f:\R^d\to \R$ if, for each $\theta\in \R^d$,
  \begin{enumerate}
    \item $\liminf_n f_n(\theta_n)\geq f(\theta)$ for any sequence $\theta_n\to \theta$
    \item $\lim_n f_n(\theta_n)=f(\theta)$ for some sequence $\theta_n\to \theta$.
  \end{enumerate}
  Thus, we say that $f$ is the epi-limit of $\crl{f_n, \ n\in \mathbb{N}}$.
\end{definition}

\begin{theorem}
    \label{thm:ermoliev-charact-local-min} Let
  \begin{enumerate}
    \item  $f:\R^d\to \R$ be locally integrable, lower bounded and lower semi-continuous, 
    \item $\crl{f_n:\R^d\to\R, \ n\in \mathbb{N}}$ be a sequence of
      differentiable functions epi-convergent to $f$. 
  \end{enumerate}
  Then  for any $\theta_* \in \mathrm{loc-}\argmin f$ there exists
  $\crl{\theta_{k} \in \R^d, \  k\in \mathbb{N}}$ such that $\theta_k\to
  \theta_*$ and $\lim_k ||\nabla f_{k}(\theta_{k})||=0$. 
\end{theorem}

\bibliographystyle{apalike}

\bibliography{complete}

\end{document}